\documentclass{amsart}
\usepackage{amsaddr}

\usepackage{amsmath,amssymb}
\usepackage{amsthm,color,verbatim}
\usepackage{booktabs}

\usepackage{color}
\definecolor{newblue}{RGB}{94,89,144}
\definecolor{newblue2}{cmyk}{1,0.6,0,0.06}
\definecolor{grey}{gray}{0.5}
\usepackage[pdfborder={0 0 0},
  colorlinks=true,
  citecolor=newblue,
  linkcolor=newblue2,
  urlcolor=newblue]{hyperref}

\usepackage{natbib}
\usepackage{graphicx}

\newtheorem{lemma}{Lemma}

\theoremstyle{remark}
\newtheorem{remark}{Remark}

\newcommand{\C}{\mathcal{C}}

\renewcommand{\P}{\mathbb{P}}
\newcommand{\R}{\mathbb{R}}
\newcommand{\W}{\mathcal{W}}
\renewcommand{\Re}{\mathrm{Re}}

\newcommand{\abs}[1]{\left| #1 \right|}
\newcommand{\dd}{\, \mathrm{d}}
\renewcommand{\vec}[1]{\boldsymbol{#1}}

\newcommand{\unit}[1]{\,\mathrm{#1}}

\newcommand{\energyError}{\mathcal{E}}
\newcommand{\constraintError}{\mathcal{Z}}

\newcommand{\liquid}{\mathrm{liquid}}
\newcommand{\agar}{\mathrm{agar}}
\newcommand{\deltaT}{\Delta t}
\newcommand{\cuticle}{\mathrm{cuticle}}

\usepackage[pagewise,mathlines,displaymath]{lineno}
\begin{document}
% \include{notes}
% \shownotes	% comment this line out to hide notes

% \begin{frontmatter}

\title[A computational method for \textit{C.\ elegans} biomechanics]{A new computational method for a model of \textit{C.\ elegans} biomechanics: Insights into elasticity and locomotion performance}
\author{Netta Cohen and Thomas Ranner}
\date{\today}

\address{School of Computing, EC Stoner Building, University of Leeds, Leeds. LS2 9JT}
\email{N.Cohen@leeds.ac.uk, T.Ranner@leeds.ac.uk}

\begin{abstract}
  An organism's ability to move freely is a fundamental behaviour in the animal kingdom.
  To understand animal locomotion requires a characterisation of the material properties, as well as the biomechanics and physiology.
  We present a biomechanical model of \textit{C.\ elegans} locomotion together with a novel finite element method.
  We formulate our model as a nonlinear initial-boundary value problem which allows the study of the dynamics of arbitrary body shapes, undulation gaits
  and the link between the animal's material properties and its performance across a range of environments.
  Our model replicates behaviours across a wide range of environments. It makes strong predictions on the viable range of the worm's Young's modulus and suggests that animals can control speed via the known mechanism of gait modulation that is observed across different media.
\end{abstract}

\maketitle

%\begin{keyword}

%\end{keyword}

%\end{frontmatter}

\section{Introduction}
\label{sec:introduction}

\subsection{Background}

Active motility and animal locomotion are fundamental for the survival of life
forms across all scales. At the microscale, cells and microorganisms most
commonly occupy fluid habitats often near surfaces. To orient
themselves and navigate in their environments such microorganisms employ
active swimming.  In this low Reynolds number regime (Reynolds number $(\Re)$ less than one), where
viscous forces dominate over inertial forces, the scallop theorem \citep{Pur77}
prohibits self-propelled locomotion for time-reversal symmetric sequences of
body postures \citep{Tay52}. It is therefore no accident that at this scale,
locomotion most commonly relies on undulations, or the propagation of waves
along the body or parts thereof, that break time-reversal symmetry
\citep{CohBoy10}.

Of particular interest in the study of undulatory locomotion and its neural
basis is the roundworm (nematode) \textit{Caenorbabditis elegans}.  The animal's invariant anatomy (959 nongonadal cells in the adult hermaphrodite of which
exactly 302 are neurons), fully annotated cell lineage, genome and nervous
system and its well characterised behaviours make it a leading model organism in
neurobiology, developmental and molecular biology.  With small size, an 
elegant periodic gait, and experimental tractability, this animal has 
also captured the interest of physicists, computer scientists and engineers,
as a model system for active swimming, neuromechanics and biorobotics. 
While in the wild, \textit{C.\ elegans} grows mostly in rotten vegetation, 
it is grown extensively
in laboratories, where it is cultured on the surface of agar gels. With a
length of $1 \unit{mm}$ and up to $2 \unit{Hz}$ undulation frequency, the worm can have a $\Re \approx
1$ in water based buffer solutions - approaching the upper bounds for a low $\Re$
treatment \citep{lauga}. On agar, as well as in liquid, \textit{C.\ elegans} move by
propagating undulatory waves from head to tail to generate forward thrust
\citep{GraLis64,Cro70,PieCheMun08,BerBoyTas09,FanWyaXie10,Lebois2012}. Body wall muscles,
lining the dorsal and ventral sides of the body generate bending.

Close inspection of \textit{C.\ elegans} crawling on the surface of an agar gel reveals a thin film of liquid on the surface and surrounding the animals. As the animal moves, surface tension presses it strongly against the gel surface \citep{Wallace1968}. To move,
the animal carves a groove in the gel along its path, and in so doing, it leaves approximately periodic tracks in its wake. As the worm moves forwards, it overcomes the gel's yield stress with its head, allowing the rest of the body to follow more easily. In contrast, motion normal to the body surface is strongly resisted. For a sufficiently stiff gel, the net result is functionally quite similar to locomotion in a solid channel \citep{FanWyaXie10,BerBoyTas09}. In contrast, in a Newtonian environment, such as buffer solution, normal and tangential forces
are of comparable magnitude. The modulation of locomotory gait as a function 
of the properties of the fluid environment has been the focus of extensive
work \citep{BerBoyTas09,FanWyaXie10,SznPurKra10,Sauvage2011,Cohen2014}, highlighting the importance of biomechanical models, as well as neurological and fluid models, of this important system \citep{Niebur1991,BoyBerCoh12,SznPurKra10,FanWyaXie10,Cohen2014,Szigeti2014,Palyanov2016,lauga}.

\subsection{Model}

In this work we present a continuum model of \textit{C.\ elegans} biomechanics and derive a new numerical method that allows for simulations of arbitrary locomotion gaits.
%  focus on presenting a new numerical method for an existing model of undulatory locomotion applied to understand forwards locomotion in \textit{C. elegans}.
The model is based on previous formulations of bending
in elastic beams given by \citet[Sec.\ 17-20]{LanLif75} and \citet{GuoMad08}
which was originally applied to the locomotion of sandfish in granular media and has also been applied to the locomotion of \textit{C.\ elegans} in Newtonian media, for example, see \citep{SznPurKra10,FanWyaXie10}.
% \notes{the sentence is too dense and the citation very obscure. I have tried to begin to unpack but more work is needed. What is the contribution of LanLif beyond the definition of elasticity? WHat is the contirbution of GuoMad?} 
% \notes{How did Sznitman use this? Not clear to me...}
% \notes{TR: my understanding is that this model was first presented by \citet{GuoMad08} based on derivations of \citet[Sec. 17-20]{LanLif75}.
%   The same equations are used by both \citet{SznPurKra10} and \citet{FanWyaXie10} and dealt with in different ways.}
% \notes{%
%   The model is Eqns. (1,2) plus the following discussion in \citet{SznPurKra10} including boundary conditions.
%   A solution is found using a given Ansatz (First eqn. in Results and Discussion).
% }

In contrast to these previous approaches,
we will use an inextensibility constraint derived from conservation of mass to close the model and thus allow the exploration of kinematic properties, including body postures predicted by the model as well as crawling speeds.
Furthermore, we derive a new presentation of the model which yields a formulation better-suited to numerical computations to solve the full nonlinear system of equations.
% A further departure of the model from similar formulations \citep{FanWyaXie10} is the choice of coordinate frame, that will yield a smooth description of the body and a more robust set of equations that lend themselves to numerical solution.

\subsection{Numerical method}

We will solve the model using a novel parametric finite element method.
In this approach the midline curve of the worm is approximated by a piecewise linear curve and the local length constraint is enforced using a piecewise constant line tension in a so-called mixed finite element method.
A semi-implicit time stepping strategy is used which produces a linear system of equation to solve at each time step.
Our approach builds on the method of \citet{DziKuwSch02} for the elasticae equations which have a global length constraint instead of our local constraint.
A key step in the method is to introduce a new variable for curvature which permits the use of piecewise linear finite elements.
A semi-discrete error bound is provided by \citet{DecDzi09}.
For models without the local length constraint, moving the nodes purely according to the model may result in very distorted mesh points which may lead to break down of the method.
\citet{BarGarNur10,BarGarNur11} introduce an artificial tangential motion which avoids this type of break down.
A separate approach using a $C^1$-splines which capture the curvature directly in a conforming method for an inextensible elastic beam has been studied by \citet{Bar13}.
Here, the local length constraint ensures that the mesh points do not become too distorted.We propose a different formulation of the constraint equation which leads to a more accurate and robust numerical scheme whilst also working with the simpler finite element spaces.
A review of different finite element approaches to geometric partial differential equations is given by \citet{Deckelnick2005}.

\subsection{Outline}

We will demonstrate the applicability of the method in simulations to
investigate biomechanical properties of the worm, to investigate body shapes
during rhythmic forward locomotion and to illustrate the capacity of the active
body to generate non-undulatory body shapes, reminiscent of omega turns, which
the worm uses to reorient \cite{Gray2005}.
Furthermore we will investigate
limits of physical properties and movement envelope of the animal through different media which will show the evolutionary trade offs involved.

The rest of the manuscript proceeds as follows.  In Section~\ref{sec:model}, we
introduce the model of \citet{GuoMad08} with our new closing assumptions. Our
focus in this section is to provide sufficient detail to see where important
biological details are included.  This section includes details of a
non-dimensionalisation which allows us to consider important parameter
relationships.  Section~\ref{sec:numerical-methods} presents the numerical
method and also provides mathematical results to show the key properties
of the numerical scheme. We conclude, in
Section~\ref{sec:simulation-results}, with numerical results to show how
the model and numerical scheme can be used to simulate important parameter
regimes.

\section{Model}
\label{sec:model}

We present the model in detail in order to demonstrate where important biological details are included which could be controlled in appropriate experimental assays.
We aim to write the model in a way that is suitable for numerical computations so work with a parametrisation which is not necessarily an arc-length parametrisation.

In \textit{C.\ elegans} and other nematodes, the body is well approximated by a
tapered cylinder. A body cavity with high hydrostatic pressure ensures that the
volume of the worm is approximately conserved. A viscoelastic cuticle
encapsulates the animal. 95 body wall muscles line the body in four quadrants.
The muscles are driven by motor neurons in the head and along the body. Muscles
are attached to the cuticle, such that their contractions induce bending. 
The neuromuscular connectivity pattern along the body restricts muscle 
contractions to the dorso-ventral plane, limiting experimental studies and 
models to date to two dimensions. To move, the animal lies on its left or
right side, propagating body waves in the dorso-ventral plane, opposite the 
direction of motion. As we focus on nematodes, we therefore assume that the
worm is positioned on its side \citep{Hart2006}.

For a sufficiently slender worm, a model of a viscoelastic beam
captures the key properties of the biomechanics and kinematics of the
locomotion embedding in fluid which vary across a wide range of viscosities \citep{FanWyaXie10} and effective
viscoelasticities \citep{BerBoyTas09,BoyBerCoh12}. 
More recent models have instead modelled a
cylindrical shell \citep{SznPurKra10} highlighting the finite aspect ratio
of the animal (length to diameter ratio $\approx 12$) with attached muscles promoting the direct actuation of
the thin cuticle. Here, we model the cuticle by a thin viscoelastic shell, and
derive the effective force equations on the midline of the body.
% We interpret the viscoelastic shell assumption loosely since we effectively couple contributions for the cuticle and the interior of the worm.
% Further work is needed to decouple these contributions \citep{Gilpin2015}.
Note that throughout, our reference to the elasticity of the cuticle is used only loosely.
Strictly speaking, our elasticity term couples contributions for the cuticle and the interior of the worm. Further work is needed to decouple these contributions \citep{Gilpin2015}.

% In nematodes in general and \textit{C.\ elegans} in particular, body waves are
% generated in the dorso-ventral plane. To move, the animal lies on its left or
% right side, propagating body waves in the plane of motion. As we focus on
% nematodes, we therefore assume that the worm is positioned on its side.

We assume that the geometry of the worm at time can be described
through a parametrisation in cylindrical geometry:
\begin{linenomath*}
\begin{align*}
  \W(t) := \{ \vec{x}(u,t) + B(0,R(u)) : u \in [0,1] \} \subset \R^3,
\end{align*}
\end{linenomath*}
% \notes{TR: the previous notation is wrong. we not only have $\W(t) \subset \R^2 \times \{0\}$. i think this is reasonable.}
%
where $\vec{x}(\cdot, t) \colon [0,1] \to \R^2 \times \{ 0 \}$ is a
smooth parametrisation of the midline $\C(t)$ of $\W(t)$, $R(u)$ is the radius of the worm at $\vec{x}(u,t)$ and $B(0,R(u))$ is a disk radius $R(u)$ in the normal plane to $\C(t)$ at $\vec{x}(u,t)$.
We assume that the radius of the worm is fixed
in time but may vary along the parametrised length of the worm $u$ and that the worm has a constant density $\varrho$.

We will choose the dimensionless variable $u$ a material coordinate which corresponds monotonically to a body coordinate. Tracking a point $\vec{x}(u,t)$ in time corresponds to following a fixed marker within the worm. Thus $u$ is an anatomically and physiologically relevant physical variable.
Here, we take $\vec{x}(0,t)$ and $\vec{x}(1,t)$ to denote the head and tail of
the animal, respectively.
The choice of a potentially non-arclength parametrisation allows us to naturally impose the forces resulting from internal pressure whilst also giving greater flexibility when working with this geometry in a computational setting.
Furthermore, the choice of parametrisation provides a much more general framework for following material coordinates in a wider variety of soft-bodied locomotion including when considering organisms which vary in length \citep{Paoletti2014}.

%For simplicity, the length of the worm (approximately 1 mm in the adult) is fixed in the model and equal to $L$.

\subsection{Geometry}

We denote by $\vec{\tau}$ the unit tangent vector, by $\vec{\nu}$ the unit
normal vector in the plane $\R^2 \times \{0\}$ and by $\vec\kappa(u,t)$ the
vector curvature of $\C(t)$.  We can define these essential geometric
quantities using the parametrisation of the midline $\vec{x}$: 
\begin{linenomath*}
\begin{subequations}
  \begin{align}
    \label{eq:tau}
    \vec\tau( u, t ) = \vec\tau(
\vec{x}( u, t ), t ) & = \frac{ \vec{x}_u( u, t ) }{\abs{ \vec{x}_u( u, t ) }
}, \quad
    \vec\nu( u, t ) = \vec\nu( \vec{x}( u, t ), t )
%    = \vec\tau( \vec{x}( u, t ), t )
    = \left( \frac{ \vec{x}_u( u, t ) }{\abs{ \vec{x}_u( u, t ) } }
\right)^\perp \\
    \label{eq:kappa}
    \vec\kappa( u, t )
    = \vec\kappa( \vec{x}( u, t ), t )
    & = \frac{1}{\abs{ \vec{x}_u( u, t ) } } \frac{\partial}{\partial u} \vec\tau( u, t )
    =\frac{1}{\abs{ \vec{x}_u( u, t ) } } \frac{\partial}{\partial u} \left(
      \frac{ \vec{x}_u( u, t ) }{\abs{ \vec{x}_u( u, t ) } }
    \right),
  \end{align}
\end{subequations}
\end{linenomath*}
where we choose the orientation $(\xi_1, \xi_2, 0 )^\perp := ( -\xi_2, \xi_1, 0 )$ for the normal vector.
%
\begin{comment}
The above differs from the more commonly adopted Frenet frame which defines the
normal as the derivative of the tangent. Our definition of the normal will
ensure that the curvature function will be one-sided and well defined even when
$\partial\vec\tau/\partial u=0$. Note that consequently the curvature in our
notation can take on negative magnitudes. Note also that for convenience, we
have adopted a convention by which the tangent vector $\vec{\tau}$ points
backwards (from head to tail).
\end{comment}
%
It can be shown that the vector curvature points in the normal direction so that if $\vec\kappa \neq 0$, we can decompose $\vec\kappa = \kappa \vec\nu$ to determine the (signed) scalar curvature $\kappa$.
% The final equation \eqref{eq:kappa} is often called the Laplace-Beltrami identity since it shows that the curvature vector is the Laplace-Beltrami operator of the position vector.
The scalar curvature sign depends on the local convexity/concavity of the worm
with respect to the choice of orientation of $\vec{\nu}$.
The relevant quantities are shown in Figure~\ref{fig:geometry}.

\begin{remark}
  In previous work on \textit{C.\ elegans} locomotion the Frenet frame has been used to define the unit normal vector field as $\vec\nu := \vec\kappa / \abs{ \vec\kappa }$ \citep{Stephens2008}.
  The above definition remains valid when $\vec\kappa = 0$ and chooses a consistent normal direction.
\end{remark}

\begin{figure}[tb]
  \centering
  \includegraphics[width=\textwidth]{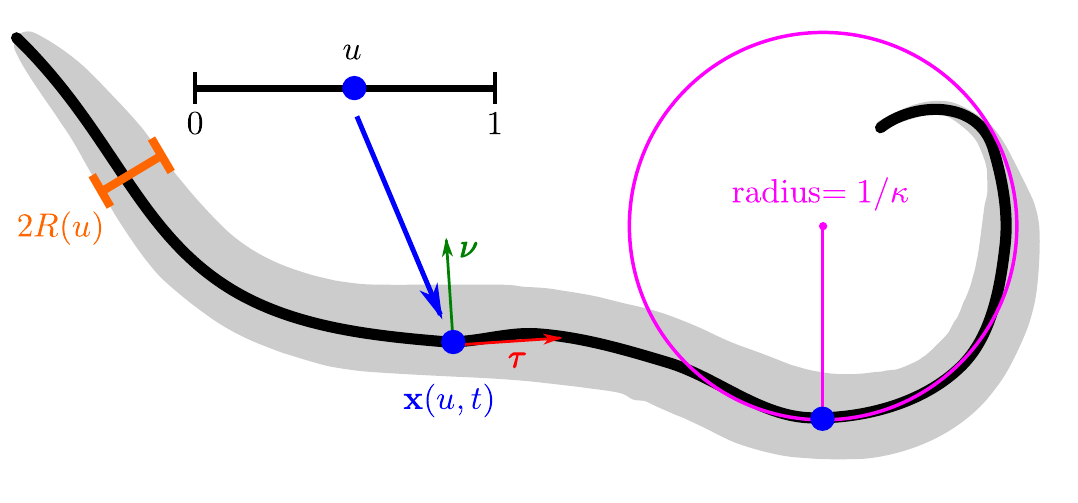}

  \caption{A sample worm posture marked up with our geometric description. The worm outline $\W$ in the plane is shown in grey and the midline $\C$ is shown in black which is parametrised by $\vec{x}$ with width $2 R$. The tangent vector $\vec{\tau}$ (red) and normal vector $\vec{\nu}$ (green) are shown. The pink osculating circle shows the radius of curvature which is the inverse of the absolute value of the scalar curvature and is equal to $\pm\abs{\kappa}$}.
  \label{fig:geometry}
\end{figure}

Although we will write our model with respect to the parameter $u$, we will also make use of an arc-length parameter $s = s(u)$ defined by
\begin{linenomath*}
\begin{align*}
  s = s(u) = \int_0^u \abs{ \vec{x}_u } \dd u.
\end{align*}
\end{linenomath*}
In particular $\partial / \partial s = \abs{\vec{x}_u}^{-1} \partial / \partial u$.
The variable $s$ represents the physical length along the body.

\subsection{Derivation of equations}

\subsubsection*{Conservation laws}

We start by considering conservation of mass and recall that we have assumed that the density of the worm, $\varrho$, is constant.
Let $\Gamma(t)$ be a portion of the worm that follows the motion of $\W(t)$ then
\begin{linenomath*}
\begin{align*}
  \frac{d}{dt} \int_{\Gamma(t)} \varrho \dd \Gamma = \varrho \frac{d}{dt} \int_{\Gamma(t)} \dd \Gamma = 0.
\end{align*}
\end{linenomath*}
% \notes{$\dd \Gamma$ is the Lebesgue measure on $\Gamma(t) \subset \R^3$ not on the curve}
%
In particular for $\Gamma(t) = \{ \vec{x}(u,t) + B( 0, R(u)) : u \in (u_1,u_2) \}$ for arbitrary $0 \le u_1 < u_2 \le 1$, we see that since $R(u)$ is fixed in time and $\varrho$ is a constant that
\begin{linenomath*}
\begin{align*}
  \frac{d}{dt} \int_{\Gamma(t)} \varrho \dd \Gamma
  = \varrho \frac{d}{dt} \int_{u_1}^{u_2} \pi R(u)^2 \abs{ \vec{x}_u } \dd u
  = \varrho \int_{u_1}^{u_2} \pi R(u)^2 \vec\tau \cdot \vec{x}_{u,t} \dd u = 0.
\end{align*}
\end{linenomath*}
Since this equation holds for all $0 \le u_1 < u_2 \le 1$ and $\varrho, R(u) > 0$, we have that
\begin{linenomath*}
\begin{align}
  \label{eq:mass-cons}
  \frac{d}{dt} \abs{ \vec{x}_{u} } = \vec\tau \cdot \vec{x}_{u,t} = 0.
\end{align}
\end{linenomath*}
We may integrate \eqref{eq:mass-cons} forwards in time to see that $\abs{\vec{x}_u}$ is invariant in time. This implies that the length of the midline curve is fixed. If the initial condition is a proportional to an arc-length parametrisation, then the parametrisation is proportional to an arc-length parametrisation at all times. We will use this in Section~\ref{sec:numerical-methods} to impose a parametrisation which is approximately proportional to arc-length. We will not directly substitute the initial values of $\abs{\vec{x}_u}$ into our geometry since we wish to be able to track the errors in our approximation.

We denote by $\vec{F}$ the internal force resultant, $\vec{M} = ( 0, 0, M )$
the internal moment resultant, and by $\vec{K}$ the external force per unit length at each cross section.
% Due to the worm's small size we neglect inertia.
Neglecting inertia \citep{FanWyaXie10,BerBoyTas09}
conservation of linear and angular momentum give the relations that \citep[Eq.\ 19.2,19.3]{LanLif75}
\begin{linenomath*}
\begin{align}
  \label{eq:cons-law}
  \vec{F}_s + \vec{K} = 0, \qquad
  \vec{M}_s + \vec\tau \times \vec{F} = 0,
\end{align}
\end{linenomath*}
We decompose $\vec{F} = p \vec{\tau} + N \vec\nu$ into transverse shear force $N$ and tangential force (line tension) $p$ and the moment $M = M^p + M^a$ into passive $M^p$ and active $M^a$ components.

\subsubsection*{Interaction with the environment}

We model the worm's interaction with the its surrounding fluid using slender body theory. In a low Reynolds number, Newtonian environment where the fluid boundaries effect are negligible, the environmental forces can be described using a Stokes system. The fluid system can be solved using the Stokeslet which reduces to normal $(K_{\vec{\nu}})$ and tangential $(K_{\vec\tau})$ drag coefficients.
This simplified formulation is also known as resistive force theory \citep{GRAY802}.
A key property of this system
is the ratio of drag coefficients, $K := K_{\vec\nu} / K_{\vec{\tau}}$.

\citet{Lig76} derived that, in water-like buffer solution, $\vec{K}$ should be given by
\begin{linenomath*}
\begin{align}
  \label{eq:rft}
  \vec{K} = -K_{\vec{\nu}} (\vec{x}_t \cdot \vec\nu) \vec\nu
  - K_{\vec{\tau}} (\vec{x}_t \cdot \vec\tau) \vec\tau, \quad
  K_{\vec{\nu}} = \frac{2 \pi \mu}{ \log(2 \alpha / R)},
  K_{\vec{\tau}} = \frac{4 \pi \mu}{ \log( 2 \alpha / R) + 0.5},
\end{align}
\end{linenomath*}
where $\mu$ is the dynamic viscosity of the fluid ($\approx 1 \, \unit{mPa} \, \unit{s}$ for water),
and $R$ is the radius of the body. The factor $\alpha$ depends on the precise waveform of undulations.
For sinusoidal and sufficiently small amplitude undulations with $\lambda = 1.5 \unit{mm}$, \citet{Lig76} showed that $\alpha \approx 0.09 \lambda$.
% Observational data \citep{BerBoyTas09} shows the variablilty of wavelength and amplitude along the body matter in \textit{C.\ elegans} undulations.
Note that the ratio $K$ is purely geometric entity in Newtonian fluids (independent of viscosity).

\begin{remark}
  Note that our definitions of $K_{\vec\nu}$ and $K_{\vec\tau}$ differ from \citet{BoyBerCoh12} by a factor of $L$ since $\vec{K}$ is force per unit length.
\end{remark}

Following 
\cite{BerBoyTas09}, we will assume that the resistive force description 
is permissible in the case of select viscoelastic (non-Newtonian) fluids.
As in this regime \eqref{eq:rft} no longer applies, we will use 
experimentally derived coefficients.

Recent works by \citet{Schulman2014} have used an extension of this theory to include the so-called wall effects from small scale swimmers moving near the boundary of the surrounding fluid.
\citet{Rabets2014} have also considered a nonlinear correction to this linear theory to account for nonlinear force-velocity relationships induced by viscoelastic effects in the surrounding fluid.
In this work, we assume that we remain in the linear regime as predicted by slender body theory and validated by \citep{BerBoyTas09} and limit our consideration to locomotion far from boundaries, so that wall effects are negligible.

\subsubsection*{Passive moment}

The passive moment is given by considering the stress over the disc $B(0,R(u))$ centered at each point along the midline:
\begin{linenomath*}
\begin{align*}
  M^p(u,t) = \int_{B(0, R(u))} \rho \sigma \dd A,
\end{align*}
\end{linenomath*}
where $\rho$ is the distance to the centre of $B(0,R(u))$, $\sigma$ is the stress, $\mathrm{d} A$ is the area element in $B(0,R(u))$.
The passive properties of the worm are given by assuming a linear Voigt model for the visco-elastic cuticle \citep{Fung1993,Linden1974} so that
\begin{linenomath*}
\begin{align*}
  \sigma = E \epsilon + \eta \epsilon_t,
\end{align*}
\end{linenomath*}
where $\epsilon$ is the strain, $E$ is the Young's modulus of the cuticle, and $\eta$ is the viscosity of the cuticle.
Since the mechanical stress arises in the cuticle, we have that $\sigma$ is zero away from the thin
cuticle.
This continuum model of the viscoelastic properties of the cuticle more
naturally
captures the forces given by the parallel springs and dampers in the articulated model of \citet{BoyBerCoh12}.

Assuming that we only allow deformations along the body (no rotational or radial deformations), we can compute from our cylindrical parametrisation that \citep[Sec.\ 17]{LanLif75}
\begin{linenomath*}
\begin{align*}
  \epsilon( u, \rho, t ) =  \rho \kappa( u, t ).
\end{align*}
\end{linenomath*}
Here, we use a reference configuration of a straight line of length $L$ and
applied the conservation of mass. We infer that the passive moment is given by
\begin{linenomath*}
\begin{align}
  \label{eq:moment}
  M^p = E I_2 \kappa + \eta I_2 \kappa_t.
\end{align}
\end{linenomath*}
Here $I_2$ is the second moment of area of the cuticle given by
\begin{linenomath*}
\begin{equation*}
I_2(u) = \int_{B(0,R(u))} \rho^2 \dd A =
\frac{\pi}{2} \left[ ( R +r_{\mathrm{cuticle}} / 2 )^4 - ( R -r_{\mathrm{cuticle}} / 2 )^4 \right],
\end{equation*}
\end{linenomath*}
where $r_{\mathrm{cuticle}}$ is the thickness of the cuticle.

\begin{remark}
  Recent experimental evidence has shown that the properties of the cuticle may not follow such a simple relationship. \citet{Backholm2015} propose that a power-law fluid model which was shown to be in good agreement with experimental results.
  %This is consistent with the study of \citet{SznPurKra10} who found a that the force-velocity curve for active muscles should have a negative slope.
  For this preliminary work, we pursue this simple model and leave more complicated constitutive relations to future work.
\end{remark}

\subsubsection*{Internal pressure}

We have modelled our worm in an idealised situation by strictly enforcing a strict inextensibility constraint (\ref{eq:mass-cons}).
In this situation, instead of using a constitutive relation for the resulting line tension of $\C$, which we denote by $p \vec{\tau}$, we treat the magnitude of the line tension as a dynamic variable and must be determined from the other equations in our model \citep{COOMER_2001}.
In this way, we see that $p \vec\tau$ may be considered as the effective force resulting from internal pressure, or mathematically a Lagrange multiplier, which enforces mass conservation.
We note that such a treatment of pressure is conventional in fluid mechanics \citep{Batchelor2000}.

\subsubsection*{Model and boundary conditions}

We consider the above equations over a fixed time interval $[0,T]$.
We assume that a smooth initial condition $\vec{x}_0$ is given which satisfies $\abs{ \vec{x}_{0,u} } \ge \gamma > 0$ for some positive constant $\gamma$.
Combining the conservation laws \eqref{eq:mass-cons} and \eqref{eq:cons-law} with the external force \eqref{eq:rft} and passive moment \eqref{eq:moment} leads to the following system of equations which we now write with respect to the parameter $u$:

Given an initial condition $\vec{x}_0$ and an active moment $M^a$,
find $\vec{x} \colon [0,1] \times (0,T) \to \R^2 \times \{ 0 \}$ and $p \colon [0,1] \times (0,T) \to \R$ such that
\begin{linenomath*}
\begin{equation}
\begin{aligned}
  \label{eq:model}
  K_{\vec{\nu}} (\vec{x}_t \cdot \vec\nu) \vec\nu
  + K_{\vec{\tau}} (\vec{x}_t \cdot \vec\tau) \vec\tau
  - \frac{1}{\abs{\vec{x}_u}} ( p \vec\tau )_u \qquad\qquad \\
  + \frac{1}{\abs{\vec{x}_u}} \left( \frac{E I_2}{\abs{\vec{x}_u}} \kappa_u \vec{\nu} + \frac{\eta I_2}{\abs{\vec{x}_u}} \kappa_{tu} \vec{\nu} + \frac{1}{\abs{\vec{x}_u}} M^a_u \vec\nu \right)_u & = 0 \\
  \vec{\tau} \cdot \vec{x}_{tu} & = 0 \\
  \vec{x}( \cdot, 0 ) & = \vec{x}_0.
\end{aligned}
\end{equation}
\end{linenomath*}
%
%\notes{there is a subscript $u$ after the bracket so we can't move $\nu$ outside the bracket}

We close the system by imposing zero moment and zero force at the two boundary ends:
\begin{linenomath*}
\begin{align*}
  - p \vec{\tau}
  + \left( \frac{E I_2}{\abs{\vec{x}_u}} \kappa_u + \frac{\eta I_2}{\abs{\vec{x}_u}} \kappa_{tu}
  + \frac{1}{\abs{\vec{x}_u}} M^a_u \right) \vec\nu & = 0 && \mbox{ for } u = 0,1, \\
  E I_2 \kappa + \eta I_2 \kappa_t
  + M^a & = 0 && \mbox{ for } u = 0,1.
\end{align*}
\end{linenomath*}

\begin{comment}
\begin{remark}
  The constraint equation in \eqref{eq:model} can be integrated forwards in time to see that $\abs{ \vec{x}_u } = \abs{ \vec{x}_{0,u} }$. This implies that if the initial condition is a proportional to an arc-length parametrisation, then the parametrisation is proportional to an arc-length parametrisation at all times.
\end{remark}
\end{comment}

\subsubsection*{Active moment}

We conclude this section with a brief discussion of the active moment.
In previous studies of worm locomotion, authors have developed integrated models which close the loop by including a nervous system which takes input from the body shape and outputs muscle deformations which lead to an active moment. 
In \textit{C.\ elegans}, body wall muscles line the cuticle (or body wall). 
We assume a smooth and uniform muscle configuration along the body such that 
the muscle midline parametrisation is identical to that of the worm's cuticle 
$u$. In \citet{BoyBerCoh12}, the muscle and cuticle lengths were considered 
as one.
In contrast, here, we allow muscle lengths to drive the bending 
dynamics of the passive cuticle.

Consider a fixed muscle configuration so that $M^a(u,t) = M^a(u)$ and suppose the body of the worm is allowed to relax to stationary posture. As time tends to infinity, we see that
\begin{linenomath*}
\begin{align*}
  M^a = -E I_2 \kappa.
\end{align*}
\end{linenomath*}
In other words, for a stationary worm,
the active moment is proportional to the resulting curvature. Indeed, we can take the same scaling and consider that $M^a$ is a preferred curvature in a similar fashion to the models of \citet{Helfrich1973} for elastic lipid bilayers.
We denote the preferred curvature arising from the muscle activation profile by $\beta$ and will use the relation $M^a = - E I_2 \beta$.

\begin{comment}
\begin{figure}[tb]
  \centering
  \includegraphics{figs/muscle}
  \caption{The different idealised muscle configurations showing the shape of the worm with a given muscle configuration at long time. Left shows $L^+ > L^-$ with positively curved midline, middle shows $L^+ = L^-$ with zero curvature and right shows $L^+ < L^-$ with negatively curved midline.}
  \label{fig:muscles}
\end{figure}

Furthermore, we can locally link the function $\beta$ with the lengths of muscles on either side of the body.
We consider a small segment of the worm so that variations in the radius are negligible in a set up as shown in Figure~\ref{fig:muscles}.
If the top muscle (in the direction $\vec\nu$) has length $L^+$ and the bottom muscle (opposite the direction $\vec\nu$) has length $L^-$, the curvature of the midline is given by
%
\begin{linenomath*}
\begin{align*}
  \beta(s) = \frac{1}{R} \frac{ L^+(s) - L^-(s) }{ L^+(s) + L^-(s) }.
\end{align*}
\end{linenomath*}
%
In particular, if $L^+ > L^-$ we have positive preferred curvature, if $L^- > L^+$ we have negative preferred curvature and if $L^+ = L^-$ we have zero curvature.
%
Thus, if we are given the muscle lengths along the body, we may compute that
%
\begin{linenomath*}
\begin{align*}
  M^a = -E I_2 \beta = -\frac{E I_2}{R} \frac{ L^+ - L^- }{ L^+ + L^- }.
\end{align*}
\end{linenomath*}
\end{comment}

Many previous experimental 
\citep{Feng2004,Cronin2005,Karbowski2006,Ramot2008,BerBoyTas09} 
%\citep{Feng2004,Cronin2005,Karbowski2006,Korta2007,Ramot2008,BerBoyTas09} 
and theoretic studies \citep{BoyBerCoh12,Butler2014,Backholm2015a} have characterised the muscle activation profile for forwards locomotion.
The undulations of the animal are well described by a travelling wave for which 
the undulation amplitude, angular frequency and wavelength all vary with 
the resistivity of the fluid environment \citep{BerBoyTas09}. For simplicity
(neglecting the variation of wavelength along the body),
we take the preferred curvature $\beta$ to be given by a periodic travelling wave of the form:
\begin{linenomath*}
\begin{align}
  \label{eq:defn-travelling-wave}
  \beta(u,t) = \beta_0(u) \Psi\left( \frac{2 \pi L u}{\lambda} - 2 \pi t \omega \right),
\end{align}
\end{linenomath*}
where $\Psi$ is a $2\pi$-periodic function, $\beta_0$ is the amplitude, $\lambda$ is the wavelength of undulations and $\omega$ is the angular frequency. Note
that here, the driving force $\beta$ follows the material coordinate $u$ (corresponding to the physiological muscle locations).

\begin{comment}
\citet{PierceShimomura2008}

\begin{verbatim}
11. Korta, J., D. A. Clark, ., A. D. Samuel. 2007. Mechanosensation and
mechanical load modulate the locomotory gait of swimming C. elegans.
J. Exp. Biol. 210:2383–2389.
12. Pierce-Shimomura, J. T., B. L. Chen, ., S. L. McIntire. 2008. Genetic
analysis of crawling and swimming locomotory patterns in C. elegans.
Proc. Natl. Acad. Sci. USA. 105:20982–20987.
13. Karbowski, J., C. J. Cronin, ., P. W. Sternberg. 2006. Conservation
rules, their breakdown, and optimality in Caenorhabditis sinusoidal
locomotion. J. Theor. Biol. 242:652–669.
16. Cronin, C. J., J. E. Mendel, ., P. W. Sternberg. 2005. An automated
system for measuring parameters of nematode sinusoidal movement.
BMC Genet. 6:5.
17. Feng, Z., C. J. Cronin, ., W. R. Schafer. 2004. An imaging system for
standardized quantitative analysis of C. elegans behavior. BMC Bioinformatics.
5:115.
18. Ramot, D., B. E. Johnson, ., M. B. Goodman. 2008. The Parallel
Worm Tracker: a platform for measuring average speed and drug induced
paralysis in nematodes. PLoS One. 3:e2208.
\end{verbatim}

\citet{Butler201} need to read in detail \url{refs/ButBraYem14.pdf}

\citet{Backholm2015a} need to read in detail \url{refs/BacKasSch15.pdf}
\end{comment}

\subsection{Non-dimensionalisation and parameter choices}
\label{sec:non-dimens-param}

In this section, we introduce characteristic values for certain parameters in order to identify key relations within the system. In this way, we translated from important parameters for individual components of the model to important parameters of the system. All used parameters are given with references in Table~\ref{tab:parameters}.

The anatomy of the worm is very well characterised and highly consistent across individuals \citep{wormbook}.
We take $L = 1 \unit{mm}$ and $R_{\mathrm{cuticle}} = 0.5 \unit{\mu m}$.
The radius of the worm is typically given by $\bar{R} = 40 \unit{\mu m}$, but in fact varies along its length. We describe this using the function
\begin{linenomath*}
\begin{align*}
  R( u ) = \bar{R} \frac{2 \big( (\varepsilon/L + u)( \varepsilon/L + 1 - u ) \big)^{\frac{1}{2}} }{ {1+2\varepsilon/L} }
\end{align*}
\end{linenomath*}
where $\varepsilon$ is a small distance chosen to regularise so that the worm always has strictly positive radius. The function is chosen to approximate the anatomical shape of the worm when the material coordinate $u$ is approximately proportional to arc-length (see below).

The material parameters of the worm are less well characterised with some estimates varying five order in magnitude. 
Of particular interest here is the choice of Young's modulus $E$ for the 
cuticle of the worm, with estimates varying from $\approx 1 \unit{kPa}$ (similar to that of sea anemonies) to
$\approx 100 \unit{MPa}$ roughly corresponding to the stiffness of rubber
\citep{Vog88,Vincent2004}.
By comparison, the elasticity of elastin (a key constituent of vertebrate 
skin and connective tissue) is estimated at $\approx 1 \unit{MPa}$. Insect 
cuticles alone, from different animals, developmental stages and conditions, 
cover a vast range of $E$ values from $\approx 10 \unit{kPa}$ in some larvae to 
even to $\approx 10 \unit{GPa}$.

The first approach to determining the elastic behaviour of the worm was by \citet{ParGooPru07}.
The authors used a piezoresistive displacement clamp and saw a linear relationship between force and displacement and that the elasticity of the cuticle shell dominates the worm's body stiffness.
They estimate a value of $E$ to be $380 \unit{MPa}$.
However, the authors also note that the stiffness of the cuticle is likely to be different in the longitudinal and circumferential directions and likely that the longitudinal bending constant is lower.
Another approach was given by \citet{SznPurKra10} who fitted parameters in a simplified version of the model presented in this paper to high speed video recordings of freely moving worms.
Based on an assumed ansatz for the muscle activation profile, they estimate that $E = 3.77 \pm 0.62 \unit{kPa}$ and $\eta = -860.2 \pm 99.4 \unit{Pa} \unit{s}$ and that these values may vary with the environmental conditions (viscosity of the surrounding fluid) \citep{Sznitman2010}.
However, these results are sensitive to the choice of muscle activation ansatz and do not decouple passive and active muscle effects.
A third approach was given by \citet{FanWyaXie10} to augment the procedure of \citet{SznPurKra10} by trapping the head of the worm in a micropipette with a second micropipette used to deflect the worm. The authors then studied the relaxation time of the worm's shape.
Using the simplified version of our model, they estimated that $E = 13 \unit{MPa}$ and $\eta < 68 \unit{kPa} \unit{s}$.
More recently, \citet{BacRyuDal13} have use a similar micropipette deflection technique with an alternative model that estimates $E = 1.3 \pm 0.3 \unit{MPa}$.
Discounting, the extreme values due to different possible side effects of the methodologies used we will take values of $E = 10 \unit{MPa}$ and $\eta = 50 \unit{kPa} \unit{s}$ (except where otherwise stated).

We consider a range of different viscous and viscoelastic fluid environments.
For Newtonian media, $K \approx 1.5$ and the viscosity values range experienced in common experimental conditions varies by more than 6 orders of magnitude \citep{FanWyaXie10}.
Parameter values for two special cases of the lower limiting case of buffer solution (denoted liquid) and
viscosity yielding similar locomotion to that found on agar are
given in Table~\ref{tab:parameters}. For non-Newtonian
media, we limit ourselves to linear viscoelastic fluids, in which
the undulations are well described by the ratio of drag coefficients 
$K_{\liquid}\le K < 100$ \citep{BerBoyTas09}, but the corresponding
values of effective drag coefficients are not well known \citep{BoyBerCoh12}.

For liquid we can use the formulae given in \eqref{eq:rft} to determine the local drag coefficients experienced by the worm. Using the viscosity of liquid ($\mu \approx 1 \unit{mPa} \unit{s}$) and typical dimensions of the worm's locomotion ($q = 0.09 \lambda = 1.35 \cdot 10^{-4} \unit{m}$, $L = 1 \unit{mm}$) and typical radius $\bar{R}$ to give
\begin{linenomath*}
\begin{align*}
  K_{\vec\nu,\liquid} = 5.2 \cdot 10^{-3} \unit{kg} \unit{m}^{-1} \unit{s}^{-1}, \qquad
  K_{\vec\tau,\liquid} = 3.3 \cdot 10^{-3} \unit{kg} \unit{m}^{-1} \unit{s}^{-1}.
\end{align*}
\end{linenomath*}
For agar, \citet{Wallace1969} estimated the tangential drag coefficients by measuring the force required to pull glass fibres of similar dimension to \textit{C.\ elegans} across the surface of the gel. This method has been augmented by \citet{Niebur1991} and \citet{BerBoyTas09} to give the coefficients we will use in this work. The process results in:
\begin{linenomath*}
\begin{align*}
  K_{\vec\nu,\agar} = 128 \unit{kg} \unit{m}^{-1} \unit{s}^{-1}, \qquad
  K_{\vec\tau,\agar} = 3.2 \unit{kg} \unit{m}^{-1} \unit{s}^{-1}.
\end{align*}
\end{linenomath*}

\begin{table}[tb]
  \centering
  \footnotesize
  \begin{tabular}{c|l|l|l}
    Parameter & Name & Typical value & Notes and references \\
    \hline
    $L$ & Worm length & $1 \unit{mm}$ & \cite{wormbook} \\
    $\bar{R}$ & Worm radius & $40 \unit{\mu m}$ & \cite{wormbook}\\
%\cite{ParGooPru07,SznPurKra10}; \\ &&& \cite{BoyBerCoh12} \\
    $r_{\cuticle}$ & Cuticle width & $0.5 \unit{\mu m}$ & \cite{ParGooPru07};\\
&&&\cite{SznPurKra10} \\
    $E$ & Young's modulus & $10 \unit{MPa}$ & \cite{ParGooPru07};\\
&&&\cite{SznPurKra10}; \\ &&& \cite{FanWyaXie10};\\&&&\cite{BacRyuDal13} \\
    $\eta$ & Cuticle viscosity & $50 \unit{kPa} \unit{s}$ & \cite{SznPurKra10};\\
&&&\cite{FanWyaXie10} \\
    $(K_{\vec\nu}, K_{\vec\tau} )_{\liquid}$ & Drag coefficients & $(5.2, 3.3) \cdot 10^{-3} \unit{kg} \unit{m}^{-1} \unit{s}^{-1}$  & \cite{Lig76};\\
    &  \hspace{1cm}  (buffer solution) && \cite{BoyBerCoh12} \\
    $(K_{\vec\nu}, K_{\vec\tau} )_{\agar}$ &  Drag coefficients (agar) & $(128, 3.2) \unit{kg} \unit{m}^{-1} \unit{s}^{-1}$ & \cite{BoyBerCoh12} \\
    $\beta_0$ & Muscle wave amplitude (agar) & $8 \unit{mm}^{-1}$ & \cite{BoyBerCoh12} \\
    $(\omega, \lambda)_{\liquid}$ & Muscle wave properties \eqref{eq:defn-travelling-wave} & $(1.76 \unit{s}^{-1}, 1.54 \unit{mm})$ & \cite{FanWyaXie10} \\
    &  \hspace{1cm}  (buffer solution) && \\
    $(\omega, \lambda)_{\agar}$ &  Muscle wave properties \eqref{eq:defn-travelling-wave} & $(0.30 \unit{s}^{-1}, 0.65 \unit{mm})$ & \cite{FanWyaXie10} \\
    & \hspace{1cm} (agar) && \\
  \end{tabular}
  \caption{List of used parameters with sources.}
  \label{tab:parameters}
\end{table}

With typical values for each of the parameters fixed, we can now rescale the equations.
We take a characteristic length scale $x_c$, characteristic time scale $t_c$ and characteristic line tension $p_c$ and introduce new non-dimensional variables by
\begin{linenomath*}
\begin{align*}
  \tilde{\vec{x}} = \vec{x} / x_c, \qquad
  \tilde{t} = t / t_c, \qquad
  \tilde{p} = p / p_c,
\end{align*}
\end{linenomath*}
Note that $u$ is already non-dimensional and these scalings induce a new scaled curvature $\tilde{\kappa} = \kappa x_c$ and final time $\tilde{T} = T / t_c$.
Since $r_{\cuticle} \ll \bar{R}$, we see that
\begin{linenomath*}
\begin{align*}
  I_2 \approx 2 \pi R^3 r_{\cuticle},
\end{align*}
\end{linenomath*}
which we rescale as
\begin{linenomath*}
\begin{align*}
  \tilde{I}_2(u) = I_2 / I_c = \frac{ 8 \big( ( \varepsilon/L + u ) ( \varepsilon/L + 1 - u ) \big)^{3/2} }{ ( 1 + 2 \varepsilon/L )^3}, \qquad \mbox{ where } I_c = 2 \pi \bar{R}^3 r_{\cuticle}.
\end{align*}
\end{linenomath*}
Note that $\tilde{I}_2(L/2)=1$.
The values in Table~\ref{tab:parameters} give that $I_c = 2.0 \cdot 10^{-19} \unit{m}^4$.
We nondimensionalise the preferred curvature due the muscle forcing as:
\begin{linenomath*}
\begin{align*}
  \tilde{\beta}(u,t) = (\beta_0 x_c) \Psi\left( \frac{2 \pi u}{\tilde{\lambda}} - 2 \pi \tilde{t} \tilde{\omega} \right),
\end{align*}
\end{linenomath*}
where $\tilde{\lambda} = \lambda / L$ and $\tilde{\omega} = \omega t_c$.

For freely moving worms on agar, we choose:
\begin{linenomath*}
\begin{align*}
  x_c = L = 1.0 \unit{mm}, \qquad
  t_c = 1/\omega = 3.3 \unit{s}, \qquad
  p_c = \frac{K_{\vec{\nu},\text{agar}} x_c^2}{t_c} =
  9.6 \cdot 10^{-7}  \unit{N}.
\end{align*}
\end{linenomath*}
This results in:
\begin{linenomath*}
\begin{align*}
  K (\tilde{\vec{x}}_{\tilde{t}} \cdot \tilde{\vec\nu}) \tilde{\vec\nu}
  + ( \tilde{\vec{x}}_{\tilde{t}} \cdot \tilde{\vec\tau}) \tilde{\vec\tau}
  - \frac{1}{\abs{\tilde{\vec{x}}_u}} ( \tilde{p} \tilde{\vec\tau} )_u \qquad\qquad \\
  +
  e \frac{1}{\abs{\tilde{\vec{x}}_u}} \left( \frac{\tilde{I}_2}{\abs{\tilde{\vec{x}}_u}} (\tilde{\kappa} - \tilde{\beta} )_u \tilde{\vec{\nu}} \right)_u
  + \tilde{\eta} \frac{1}{\abs{\tilde{\vec{x}}_u}} \left( \frac{\tilde{I}_2}{ \abs{\tilde{\vec{x}}_u}} \tilde{\kappa}_{\tilde{t}u} \tilde{\vec{\nu}} \right)_u
  & = 0 \\
  \left. - \tilde{p} \tilde{\vec{\tau}}
  + e \left( \frac{\tilde{I}_2}{\abs{\tilde{\vec{x}}_u}} ( \tilde{\kappa}-\tilde{\beta})_u \right) \tilde{\vec\nu}
  \right|_{u=0,1} & = 0 \\
  \left. e \tilde{I_2} (\tilde\kappa - \tilde\beta) \right|_{u=0,1} & = 0 \\
  & = 0 \\
  \tilde{\vec{\tau}} \cdot \frac{1}{\abs{\tilde{\vec{x}}_u}} \tilde{\vec{x}}_{\tilde{t}u} & = 0 \\
  \tilde{\vec{x}}( \cdot, 0 ) & = \tilde{\vec{x}}_0,
\end{align*}
\end{linenomath*}
where
\begin{linenomath*}
\begin{align*}
  K := K_{\vec{\nu},\mathrm{agar}} / K_{\vec{\tau},\mathrm{agar}}
  =  39,
  \qquad
  e := \frac{E I_c t_c}{x_c^4 K_{\vec{\tau},\mathrm{agar}}}
  = 1.0 \\
  \qquad \mbox{ and } \qquad
  \tilde\eta := \frac{\eta I_c}{x_c^4 K_{\vec{\tau},\mathrm{agar}}}
  = 1.6 \cdot 10^{-3} \ll 1.
\end{align*}
\end{linenomath*}

The results in agar are consistent with findings of
\citet{Niebur1991,FanWyaXie10} and \citet{BoyBerCoh12} that internal viscosity of the worm
is negligible for locomotion in sufficiently resistive environments. We remark,
however, that by considering the worm's locomotion in buffer solution we would
recover $\tilde{\eta} \approx 1$ and so internal viscosity is important in this
case. Motivated by this, in this study, we focus our attention on conditions
where $\tilde\eta \ll 1$. By neglecting viscous contributions, we are able to
focus on contributions of internal elastic forces on the dynamics of the body.  

In what follows, we proceed with the non-dimensional formulation, but revert 
to original variable labels. The resulting system we consider is as follows.
Given an initial condition $\vec{x}_0$ and an active muscle moment $\beta$,
find $\vec{x} \colon [0,1] \times [0,T] \to \R^2 \times \{ 0 \}$ and $p \colon [0,1] \times [0,T] \to \R$ such that
\begin{linenomath*}
\begin{equation}
\begin{aligned}
  \label{eq:nondim-model}
  K (\vec{x}_t \cdot \vec\nu) \vec\nu
  + (\vec{x}_t \cdot \vec\tau) \vec\tau
  - \frac{1}{\abs{\vec{x}_u}} ( p \vec\tau )_u \qquad\qquad \\
  + \frac{e}{\abs{\vec{x}_u}} \left( \frac{I_2}{\abs{\vec{x}_u}} (\kappa-\beta)_u \vec{\nu} \right)_u & = 0 \\
  \vec{\tau} \cdot \vec{x}_{tu} & = 0 \\
  \left. - p \vec{\tau}
  + e \left( \frac{I_2}{\abs{\vec{x}_u}} (\kappa-\beta)_u \right) \vec\nu
  \right|_{u=0,1}
  & = 0 \\
  \left. e I_2 (\kappa -\beta) \right|_{u=0,1} & = 0 \\
   \vec{x}( \cdot, 0 ) & = \vec{x}_0.
\end{aligned}
\end{equation}
\end{linenomath*}

\section{Numerical methods}
\label{sec:numerical-methods}

The key idea of the numerical method is to first the equations in a suitable weak formulation which results in three variables to solve for: parametrisation $\vec{x}$, curvature difference $\vec{y} = \vec\kappa - \beta\vec\nu$ and pressure $p$.
The parametrisation and curvature difference are transferred to discrete representation by their values at the nodes of a mesh, with piecewise linear interpolation between, and pressure is transferred to a piecewise constant variable.
The time discretisation is similar to a first-order backward Euler discretisation, however we treat the geometric terms explicitly (i.e.\ using the values from the previous time step).
The result is a system of linear equations (saddle point problem) to solve at each time step where the solution approximates the solution of the nonlinear model.

\subsection{Weak formulation}

We take a smooth test function $\vec\phi \colon [0,1] \to \R^2 \times \{ 0 \}$ and take the scalar product of \eqref{eq:nondim-model} with $\vec\phi$ and integrate with respect to the measure $(\abs{ \vec{x}_u } \dd u)$ to see
\begin{linenomath*}
\begin{multline*}
  \int_0^1 ( K \vec{x}_t \cdot \vec\nu \vec\nu + \vec{x}_t \cdot \vec\tau \vec\tau ) \cdot \vec\phi \abs{ \vec{x}_u } \dd u
  - \int_0^1 ( p \vec\tau )_u \cdot \vec\phi \dd u \\
  + e \int_0^1 \left( \frac{I_2}{\abs{ \vec{x}_u }} (\kappa-\beta)_u \vec\nu \right)_u \cdot \vec\phi \dd u = 0
\end{multline*}
\end{linenomath*}
We integrate by parts all but the first term (the boundary terms disappear due to the boundary condition):
\begin{linenomath*}
\begin{multline}
  \label{eq:weak-form}
  \int_0^1 ( K \vec{x}_t \cdot \vec\nu \vec\nu + \vec{x}_t \cdot \vec\tau \vec\tau)
  \cdot \vec\phi \abs{ \vec{x}_u } \dd u
  + \int_0^1 p \vec\tau \cdot \vec\phi_u \dd u \\
  - e \int_0^1 \frac{I_2}{\abs{ \vec{x}_u }} (\kappa - \beta)_u \vec\nu \cdot \vec\phi_u \dd u = 0
\end{multline}
\end{linenomath*}
In addition, for the incompressible constraint we multiply by a smooth function $q$:
\begin{linenomath*}
\begin{align}
  \label{eq:weak-incomp}
  \int_0^1 q \vec{\tau} \cdot \vec{x}_{tu} \dd u = 0.
\end{align}
\end{linenomath*}

The formulation of \eqref{eq:weak-form} and \eqref{eq:weak-incomp} takes an initial condition for $\vec{x}$ and active moment $\beta$ as data. The problem is to find the two unknowns $\vec{x}$ and $p$ from the two equations such that these equations hold for all smooth functions $\vec{\phi}$ and $q$.

\subsection{Semi-discrete finite element method}

We discretise in space using a mixed $P_1$-$P_0$ mixed finite element method.
The key ideas for the numerical scheme is to augment a parametric finite element method which uses piecewise linear functions for elasticae proposed by \citet{DziKuwSch02}, analysed by \citet{DecDzi09}, with piecewise constant functions for the pressure term. We believe this approach is well suited due to the saddle point nature of the equations.
We denote spatially discrete quantities with a subscript $h$ which relates to the mesh spacing.

We divide the parameter domain $(0,1)$ into $N$ intervals given by $e_j = [ u_j, u_{j+1} ]$ for $j=0,\ldots,N-1$ where $0 = u_0 < u_1 < \ldots < u_N = 1$. We denote by $h_j = u_{j+1} - u_j$ the element size of each element and $h = \max_{j} h_j$.
The space $V_h$ is the space of continuous vector valued piecewise linear finite element functions:
\begin{linenomath*}
\begin{align*}
  V_h := \{ \vec{\phi}_h \in C(0,1 ; \R^2 \times \{0\} ) : \vec{\phi}_h |_{e_j} \mbox{ is affine, for } j = 0,\ldots,N-1 \}.
\end{align*}
\end{linenomath*}
Here, $\vec{\phi}_h |_{e_j}$ denotes the restriction of $\vec{\phi}_h$ to the domain $e_j = [u_j, u_{j+1}]$.
We denote by $I_h$ the standard interpolation operator from continuous functions into $V_h$.
We also use the notation $V_{h,0}$ for the subspace of finite element functions in $V_h$ which take the value zero at $u=0,1$.
The space $Q_h$ is the space of scalar valued piecewise constant functions:
\begin{linenomath*}
\begin{align*}
  Q_h := \{ q_h \in L^2( 0, 1 ) : q_h |_{e_j} \mbox{ is constant for } j = 0, \ldots,N-1 \}.
\end{align*}
\end{linenomath*}

For a discrete parametrisation $\vec{x}_h \in V_h$, we can take an element-wise derivative of $\vec{x}_h$ to define the discrete tangent vector given by $\vec\tau_h := \vec{x}_{h,u} / \abs{ \vec{x}_{h,u} }$. This is a piecewise constant vector valued function. In our semi-discrete scheme, we use $\vec\tau_h$ in order to define local tangential and normal projections $\P_\tau^h$ and $\P_\nu^h$ given by
\begin{linenomath*}
\begin{align*}
  \P_\tau^h \vec\xi = ( \vec\xi \cdot \vec\tau_h ) \vec\tau_h, \qquad
  \P_\nu^h \vec\xi = \vec\xi - ( \vec\xi \cdot \vec\tau_h ) \vec\tau_h.
\end{align*}
\end{linenomath*}
We also require a discrete normal in order to compute the muscle forcing terms. In fact, we will require a discrete normal at each node, hence we will compute with the mean of the normal vector on each side of the node:
\begin{linenomath*}
\begin{align*}
  \tilde{\vec\nu}_h( u_j ) =
  \begin{cases}
    \vec\tau_h^\perp(e_0) & \mbox{ for } j = 0 \\
    \frac{1}{2} \left( \vec\tau_h^\perp(e_j) + \vec\tau_h^\perp(e_{j+1} ) \right) & \mbox{ for } j = 1, \ldots, N-1 \\
    \vec\tau_h^\perp(e_{N}) & \mbox{ for } j = N.
  \end{cases}
\end{align*}
\end{linenomath*}

Note this it is not possible to directly take a derivative of $\vec{\tau}_h$ in order to define a discrete curvature vector. Instead we use a weak form. Multiplying the Laplace-Beltrami identity \eqref{eq:kappa} by a smooth function $\vec\psi$ which is zero at 0 and 1 gives:
\begin{linenomath*}
\begin{align*}
  \int_0^1 \vec\kappa \cdot \vec\psi \abs{ \vec{x}_u } \dd u
  = - \int_0^1 \frac{\vec{x}_u}{\abs{\vec{x}_u}} \cdot \vec\psi_u \dd u.
\end{align*}
\end{linenomath*}
We will use a new variable $\vec{y} = \vec\kappa - \beta \vec\nu$. This results in a second-order splitting scheme \citep{EllFreMil89} which allows the use of non-conforming piecewise linear finite element functions.
We also use mass lumping \citep{Tho84} in order to further simplify the resulting system of equations.
We use the notations
\begin{linenomath*}
\begin{align}
  \label{eq:mass-lumping}
  ( \vec{\xi}, \vec\eta )_h = I_h( \vec{\xi} \cdot \vec\eta ), \qquad
  \abs{ \vec\xi }_h^2 = ( \vec{\xi}, \vec{\xi} )_h.
\end{align}
\end{linenomath*}

We assume we are given an initial condition $\vec{x}_{h,0} \in V_h$ that satisfies $\abs{ \vec{x}_{h,0} } \ge \gamma > 0$ for some positive constant $\gamma$. In the analysis, we take an initial value for $\vec{y}_{h,0} \in V_{h,0}$ given as the solution of
\begin{linenomath*}
\begin{align*}
  \int_0^1 ( \vec{y}_{h,0}, \vec{\psi}_h )_h \abs{ \vec{x}_{h,0,u} } \dd u
    + \int_0^1 \frac{1}{\abs{\vec{x}_{h,0,u}}} \vec{x}_{h,0,u} \cdot \vec{\psi}_{h,u} \dd u
    & \\
    \nonumber
    + \int_0^1 ( \beta \tilde{\vec\nu}_h, \vec\psi_h )_h \abs{ \vec{x}_{h,0,u} } \dd u
    & = 0
    && \mbox{ for all } \vec\psi_h \in V_{h,0}.
\end{align*}
\end{linenomath*}

The semi-discrete scheme is:
Given $\vec{x}_{h,0}$ and a muscle activation forcing function $\beta$, find for all $t \in [0,T]$, $\vec{x}_h(\cdot,t) \in V_h$, $\vec{y}_h(\cdot,t) \in V_{h,0}$ and $p_h(\cdot,t) \in Q_h$ such that
\begin{linenomath*}
\begin{subequations}
  \label{eq:semi-discrete}
  \begin{align}
    \label{eq:semi-discrete-a}
    \int_0^1 ( K \P_\nu^h + \P_\tau^h ) ( \vec{x}_{h,t}, \vec{\phi}_h )_h \abs{ \vec{x}_{h,u} } \dd u
    & \\
    \nonumber
    + \int_0^1 p_h \vec\tau_h \cdot \vec\phi_{h,u} \dd u
    - e \int_0^1 \frac{I_2}{\abs{\vec{x}_{h,u} }} \P_\nu^h \vec{y}_{h,u} \cdot \vec{\phi}_{h,u}
    & = 0
    && \mbox{ for all } \vec\phi_h \in V_h \\
    \label{eq:semi-discrete-b}
    \int_0^1 ( \vec{y}_h, \vec{\psi}_h )_h \abs{ \vec{x}_{h,u} } \dd u
    + \int_0^1 \frac{1}{\abs{\vec{x}_{h,u} }} \vec{x}_{h,u} \cdot \vec{\psi}_{h,u} \dd u
    & \\
    \nonumber
    + \int_0^1 ( \beta \tilde{\vec\nu}_h, \vec\psi_h )_h \abs{ \vec{x}_{h,u} } \dd u
    & = 0
    && \mbox{ for all } \vec\psi_h \in V_{h,0} \\
    \label{eq:semi-discrete-c}
    \int_0^1 \vec{\tau}_h \cdot \vec{x}_{h,u,t} q_h \dd u
    & = 0
    && \mbox{ for all } q_h \in Q_h.
  \end{align}
\end{subequations}
\end{linenomath*}

We immediately see that the length constraint is satisfied. Indeed,
testing (\ref{eq:semi-discrete-c}) with $q_h = \chi_{e_j} / h_j$ the characteristic function of $e_j$ weighted by the inverse of its length,for all $j = 1, \ldots, N-1$, we see that
\begin{linenomath*}
\begin{align}
  \abs{ \vec{x}_{h,u} }_t
  = \int_0^1 \vec\tau_h \cdot \vec{x}_{h,u,t} \chi_{e_j} / h_j \dd u
  = 0.
\end{align}
\end{linenomath*}
Hence, $\abs{\vec{x}_{h,u}}_t \equiv 0$ everywhere in space and time.

\begin{lemma}
  Assume that $I_2 = 1$ is constant in space and time, then any solution $(\vec{x}_h, \vec{y}_h, p_h)$ satisfies the stability estimate
  \begin{linenomath*}
  \begin{multline}
    \sup_{t \in [0,T]} \frac{e}{2} \int_0^1 \abs{ \vec{y}_h }_h^2 \abs{ \vec{x}_{h,u} } \dd u
    + \int_0^T \int_0^1 \abs{ \vec{x}_{h,t} }_h^2 \abs{ \vec{x}_{h,u} } \dd u \dd t \\
    \le \frac{e}{2} \int_0^1 \abs{ \vec{y}_{h,0} }^2 \abs{ \vec{x}_{h,u} } \dd u + c e \int_0^T \int_0^1 \abs{ (\beta \tilde{\vec{\nu}}_h )_t }_h^2 \abs{ \vec{x}_{h,u} } \dd u \dd t.
  \end{multline}
\end{linenomath*}
\end{lemma}

\begin{proof}
  We start by testing \eqref{eq:semi-discrete-a} with $\vec\phi_h = \vec{x}_{h,t}$ and sum with the result of testing \eqref{eq:semi-discrete-c} with $q_h = -p_h$ to see
  \begin{linenomath*}
  \begin{align*}
    \int_0^1 ( K \P_\nu + \P_\tau) ( \vec{x}_{h,t}, \vec{x}_{h,t} )_h \abs{ \vec{x}_{h,u} } \dd u
    - e \int_0^1 \frac{1}{\vec{x}_{h,u}} \P_\nu^h \vec{y}_{h,u} \vec{x}_{h,t,u} \dd u = 0.
  \end{align*}
\end{linenomath*}

  Next, we take a time derivative of \eqref{eq:semi-discrete-b} and then test with $\vec\psi_h = \vec{y}_h$, and apply the identity $\abs{\vec{x}_{h,u}}_t = 0$, to see
  \begin{linenomath*}
  \begin{multline*}
    \int_0^1 ( \vec{y}_h, \vec{y}_{h,t} )_h \abs{ \vec{x}_{h,u} } \dd u + \int_0^1 \frac{1}{\abs{ \vec{x}_{h,u} }} \P_\nu^h \vec{x}_{h,u,t} \cdot \vec{y}_{h,u} \dd u \\
    + \int_0^1 ( ( \beta \tilde{\vec{\nu}}_h )_t, \vec{y}_h )_h \abs{ \vec{x}_{h,u} } \dd u
    = 0.
  \end{multline*}
\end{linenomath*}
  Combining the two previous equations, along the with length constraint gives:
  \begin{linenomath*}
  \begin{multline*}
    \frac{e}{2} \frac{d}{dt} \int_0^1 \abs{ \vec{y}_h }_h^2 \abs{ \vec{x}_{h,u} } \dd u
    + \int_0^1 ( K \P_\nu^h + \P_\tau^h ) \abs{ \vec{x}_{h,t} }_h^2 \abs{ \vec{x}_{h,u} } \dd u \\
    = - e \int_0^1 ( ( \beta \tilde{\vec{\nu}}_h )_t, \vec{y}_h )_h \abs{ \vec{x}_{h,u} } \dd u.
  \end{multline*}
\end{linenomath*}
  The result then follows using a Young's inequality, the fact that $K \ge 1$ and a Gr\"{o}nwall inequality.
\end{proof}

\subsection{Time discretisation}

We use a semi-implicit time discretisation which results in a linear system of equations to solve at each time step.
This idea was first suggested in the context of geometric partial differential equations by \cite{Dzi90}.
The key idea is that geometric terms should be treated explicitly.
Throughout this section we will assume, for simplicity, that $\abs{\vec{x}_{0,h,u}} = 1$.

For simplicity, we set a fixed time step $\deltaT$ and partition the time interval by $0 = t^0 < t^1 < \ldots < t^M = T$ where $t^{k+1} - t^k = \deltaT$ for $k = 0,1,\ldots, M-1$.
We propose the following time discretisation:
Given $\vec{x}_{h,0} \in V_h$ and a muscle forcing $\vec{\beta}$, for $k = 0,1,\ldots, M$, find $\vec{x}_h^{k+1} \in V_h$, $\vec{y}_h^{k+1} \in V_{h,0}$ and $p_h^{k+1} \in Q_h$ such that
\begin{linenomath*}
\begin{subequations}
  \begin{align}
    \int_0^1 ( K \P_\nu^k + \P_\tau^k ) \left( \frac{\vec{x}^{k+1} - \vec{x}^{k}}{\deltaT}, \vec{\phi}_h \right)_h \abs{ \vec{x}_u^k } \dd u
    & \\
    \nonumber
    + \int_0^1 p_h^{k+1} \vec\tau_h^k \cdot \vec\phi_{h,u} \dd u
    - e \int_0^1 \frac{I_2}{\abs{\vec{x}_{h,u}^k }} \P_\nu^k \vec{y}_{h,u}^{k+1} \cdot \vec{\phi}_{h,u}
    & = 0
    && \mbox{ for all } \vec\phi_h \in V_h \\
    \int_0^1 \left( \vec{y}_h^{k+1}, \vec{\psi}_h \right)_h \abs{ \vec{x}_{h,u}^k } \dd u
    + \int_0^1 \frac{1}{\abs{\vec{x}_u^k}} \vec{x}_{h,u}^{k+1} \cdot \vec{\psi}_{h,u} \dd u
    & \\
    \nonumber
    + \int_0^1 ( \beta(t^{k+1}) \tilde{\vec\nu}_h^k, \vec\psi_h )_h \abs{ \vec{x}_{h,u}^k } \dd u
    & = 0
    && \mbox{ for all } \vec\psi_h \in V_{h,0} \\
    \label{eq:constraint-hk}
    \int_0^1 ( \vec{\tau}_h^k \cdot \vec{x}_{h,u}^{k+1} - 1 ) q_h \dd u
    & =  0
    && \mbox{ for all } q_h \in Q_h.
  \end{align}
\end{subequations}
\end{linenomath*}

In addition to the usual time discretisation, we have chosen to integrate the constraint equation forwards in time.
This gives us more control over the length element $\mu^k = \abs{\vec{x}_{h,u}^k}$ as shown in the following Lemma.

\begin{lemma}
  If there exists a solution such that $\abs{ \vec{\tau}_h^{k+1} - \vec{\tau}_h^k } < 2$, then
  \begin{linenomath*}
  \begin{align}
    1 \le \mu^{k+1}
    = \frac{ 1 }{ 1 - \frac{1}{2} \abs{ \vec{\tau}_h^{k+1} - \vec{\tau}_h^{k} }^2 }.
  \end{align}
\end{linenomath*}
\end{lemma}

\begin{proof}
Equation \eqref{eq:constraint-hk} with $q_h = \chi_{e_j}$, the characteristic function of $e_j$, gives the element-wise identity
\begin{linenomath*}
\begin{align*}
  \vec{\tau}_h^k \cdot \vec{x}_{h,u}^{k+1} = 1
\end{align*}
\end{linenomath*}
Then using the definition of $\mu^{k+1}$, we have
\begin{linenomath*}
\begin{align*}
  \mu^{k+1} & = \vec{\tau}_h^{k+1} \cdot \vec{x}_{h,u}^{k+1}
  = ( \vec{\tau}_h^{k+1} - \vec{\tau}_h^{k} ) \cdot \vec{x}_{h,u}^{k+1} + 1 \\
  & = \mu^{k+1} ( 1 - \vec{\tau}_h^{k+1} \cdot \vec{\tau}_h^k ) + 1
  = \frac{1}{2} \mu^{k+1} \abs{ \vec{\tau}_h^{k+1} - \vec{\tau}_h^{k} }^2 + 1.
\end{align*}
\end{linenomath*}
Since $\frac{1}{2} \mu^{k+1} \abs{ \vec{\tau}_h^{k+1} - \vec{\tau}_h^{k} }^2 \ge 0$, we have $\mu^{k+1} \ge 1$.
Furthermore if $\abs{ \vec{\tau}_h^{k+1} - \vec{\tau}_h^{k} }^2 < 2$, this equation can be rearranged to see the desired result.
\end{proof}

\begin{remark}
  An alternative scheme could replace \eqref{eq:constraint-hk} with
  \begin{linenomath*}
  \begin{align}
    \label{eq:constraint-hk'} \tag{\ref{eq:constraint-hk}'}
    \int_0^1 q_h \vec{\tau}_h^k \cdot \left( \vec{x}_{h,u}^{k+1} - \vec{x}_{h,u}^{k} \right) \dd u & = 0
  && \mbox{ for all } q_h \in Q_h.
  \end{align}
\end{linenomath*}
  In this case, similar calculations would give that
  \begin{linenomath*}
  \begin{align*}
    \mu^{k+1} = \frac{ \mu^k }{ 1 - \frac{1}{2} \abs{ \vec{\tau}_h^{k+1} - \vec{\tau}_h^{k} }^2 }.
  \end{align*}
  \end{linenomath*}
  This would imply that $( \mu^{k} )$ is increasing and this is strict unless subsequent solutions are equal. In general, using \eqref{eq:constraint-hk'}
results in an increasing error in the constraint equation and potentially unstable solution. However, the error will be small in the special (trivial) case with no muscle forcing $\beta = 0$.
\end{remark}

\subsection{Matrix formulation}

The finite element formulation allows the above analytical results, however the precise implementation of the method can be more clearly seen when written in matrix form.
First, we introduce a nodal basis $\{ \vec{e}_d \phi_j\}_{d=1,2, j=0,N+1}$ of $V_h$ given by $\phi_j( u_i ) = \delta_{ij}$ (the standard hat functions) and $\{ \vec{e}_d \}$ the standard unit basis of $\R^2$.

We decompose $\vec{x}_h^k, \vec{y}_h^k$ and $p_h^k$ by
\begin{linenomath*}
\begin{align*}
  \vec{x}^k_h( u ) = \sum_{j=0}^{N} \vec{x}^k_j \phi_j(u), \quad
  \vec{y}^k_h( u ) = \sum_{j=1}^{N-1} \vec{y}^k_j \phi_j(u), \quad
  p_h^k( u ) = \sum_{j=0}^{N-1} p^k_j \chi_{e_j}( u ).
\end{align*}
\end{linenomath*}
To write down the scheme as a system of linear equations of the nodal values $\{ \vec{x}_j^k \}, \{ \vec{y}_j^k \}$ and $\{ p_j^k \}$,
we introduce the element wise variables:
\begin{linenomath*}
\begin{align*}
  q^k_i = \abs{ \vec{x}^k_{i+1} - \vec{x}^k_{i} }, &\qquad
  \vec\tau^k_i = \vec{\tau}_h^k( e_i ) \\
  \P^k_i \vec\xi = \vec\xi - ( \vec\tau^k_i \cdot \vec\xi ) \vec\tau^k_i, &\qquad
  \mathbb{D}^k_i \vec\xi = K \vec\xi + ( 1 - K ) ( \vec\tau^k_i \cdot \vec\xi ) \vec\tau^k_i,
\end{align*}
\end{linenomath*}
and further denote by $u_{j \pm 1/2} = ( u_j + u_{j\pm1} ) / 2$. Then given a forcing function $\beta$ and a previous solution $\{ \vec{x}_i^k \}$, we compute $\{ q^k_i \}, \{ \vec\tau^k_i \}, \{ \P^k_i \}$ and $\mathbb{D}^k_i$ and find $\{ \vec{x}_i^{k+1} \}, \{ \vec{y}_i^{k+1} \}$ and $\{ p_i^{k+1} \}$ as the solution of the linear system:
\begin{linenomath*}
\begin{align*}
  \left( q_0^k \mathbb{D}^k_{0} \right) \frac{ \vec{x}_0^{k+1} - \vec{x}_0^{k} }{ \Delta t }
  + \left( -\vec\tau^k_{0} p_0^{k+1} \right)
  + e \left( I_2( u_{1/2} ) \P^k_{0} \frac{ \vec{y}_{1}^{k+1} - \vec{y}_{0}^{k+1} }{q_0^k} \right)
  = 0 \qquad \\
  \left( q_{i-1}^k \mathbb{D}^k_{i-1} + q_i^k \mathbb{D}^k_{i} \right) \frac{ \vec{x}_i^{k+1} - \vec{x}_i^{k} }{ \Delta t }
  + \left( \vec\tau^k_{i-1} p_{i-1}^{k+1} - \vec\tau^k_{i} p_i^{k+1} \right) \qquad\qquad\qquad\qquad \\
  + e \left( I_2( u_{i-1/2} ) \P^k_{i-1} \frac{ \vec{y}_{i-1}^{k+1} - \vec{y}_{i}^{k+1} }{q_{i-1}^k}
  + I_2( u_{i+1/2} ) \P^k_{i} \frac{ \vec{y}_{i+1}^{k+1} - \vec{y}_{i}^{k+1} }{q_i^k} \right)
  = 0 \qquad \\
  \mbox{ for } i = 1, \ldots, N-1 \\
  \left( q_{N-1}^k \mathbb{D}^k_{N-1} \right) \frac{ \vec{x}_N^{k+1} - \vec{x}_N^{k} }{ \Delta t }
    + \left( \vec\tau^k_{N-1} p_{N-1}^{k+1} \right) \qquad\qquad\qquad\qquad \\
  + e \left( I_2( u_{N-1/2} ) \P^k_{N-1} \frac{ \vec{y}_{N-1}^{k+1} - \vec{y}_{N}^{k+1} }{q_{N-1}^k}
  \right)
  = 0 \qquad \\
  \vec{y}_0^{k+1} + \beta( u_0, t^{k+1} ) \tilde{\vec\nu}_0^k = 0 \qquad \\
  \left( q_{i-1}^k + q_i^k \right) \vec{y}_i^{k+1}
  + \left( \frac{\vec{x}_{i-1}^{k+1} - \vec{x}_i^{k+1}}{q_{i-1}^k} + \frac{\vec{x}_{i+1}^{k+1} - \vec{x}_i^{k+1}}{q_{i}^k} \right) \qquad\qquad\qquad\qquad \\
    + \left( q_{i-1}^k + q_i^k \right) \beta( u_i, t^k ) \tilde{\vec\nu}^k_i
    = 0 \qquad \\
  \mbox{ for } i = 1, \ldots, N-1 \\
  \vec{y}_N^{k+1} + \beta( u_N, t^{k+1} ) \tilde{\vec\nu}_N^k = 0 \qquad \\
  \frac{ \vec{x}_{i+1}^k - \vec{x}_i^k }{q_i^k} \cdot ( \vec{x}_{i+1}^{k+1} - \vec{x}_i^{k+1} ) - 1 = 0 \qquad \\
  \mbox{ for } i = 0, \ldots, N-1. \\
\end{align*}
\end{linenomath*}

The resulting matrix system is solved by a direct sparse solver. For the numerical experiments presented in this paper we use the software package \textsf{UMFPACK} \citep{umfpack}.

\section{Results}
\label{sec:simulation-results}

In what follows, we first demonstrate the numerical stability of the model,
before applying the model to address questions about \textit{C.\ elegans} 
material properties and its forward locomotion.
Here, we compute in nondimensional form but show results with respect to the
dimensioned variables. In particular $\bar\eta = 0$, i.e.  body damping due to
internal viscosity is neglected.
In all simulations our initial condition is taken to be proportional to arc-length and use $\deltaT = 10^{-3}$ and $128$ points along the body.

\subsection{Model validation}

\subsubsection*{Convergence tests demonstrate good error bounds for long-time
simulations}

To show the numerical properties of our scheme we take a test problem where $\beta$ is given by a travelling sine-wave:
\begin{linenomath*}
\begin{align*}
  \beta( u, t ) = \pi \sin( 1.5 \pi u - 2 \pi t ).
\end{align*}
\end{linenomath*}
For the purpose of validating the numerical scheme, we restrict ourselves to the parameters $K = e = I_2 = 1$.
We take the initial condition to be a straight line: $\vec{x}_0(u) = (u,0)$ and run until $T = 10$.
We refer to the implementation using \eqref{eq:constraint-hk} to be the measure constraint and \eqref{eq:constraint-hk'} to be the velocity constraint.
We simulate our model with $\deltaT = 10^{-1}, 10^{-2}, 10^{-3}, 10^{-4}$ and $h = 2^{-16}, 2^{-32}, 2^{-64}, 2^{-128}$ for both forms of the constraint equation.

We track over time the energy $\mathcal{E}( \vec{x}_h, \vec{y}_h )$ given by
\begin{linenomath*}
\begin{align*}
  \energyError^k = \energyError( \vec{x}_h^k, \vec{y}_h^k ) = \int_0^1 \abs{ \vec{y}_h^k }_h^2 \abs{ \vec{x}_{h,u}^k } \dd u,
\end{align*}
\end{linenomath*}
(see \eqref{eq:mass-lumping} for the definition of lumped absolute value $\abs{\cdot}_h$)
and the error in the constraint equation:
\begin{linenomath*}
\begin{align*}
  \constraintError^k = \max_{(0,1)} \abs{ 1 - \mu^k }.
\end{align*}
\end{linenomath*}
Results are shown in Tables~\ref{tab:conv} and Figure~\ref{fig:conv}.
We first see that the convergence in the time step $\Delta t$ is much slower than in the spatial step $h$. In particular, we see large errors in the constraint equation for the velocity constraint approach. This is consistent with our analysis.
By choosing to apply the constraint to the measure instead of the velocity we see that not only is the constraint more accurately enforced but the error in the shape of the worm's body is also lower. This implies long simulations can be run stably and that the errors in body shape and thrust remain bounded.

In simulations starting from a straight line configuration, there is a short time where the body readjusts towards the usual periodic gait. In this time, the body is often rotated. However, after this initial transient we do not see any further rotations.
As a result in all subsequent calculations we will adjust the coordinate system to remove this transient rotation.

\begin{table}[tbh]
  \centering
  \begin{tabular}{c|cc|cc}
  & \multicolumn{2}{c|}{Measure constraint} & \multicolumn{2}{|c}{Velocity constraint} \\
  $\Delta t$ & $\energyError^M$ & $\constraintError^M$ & $\energyError^M$ & $\constraintError^M$ \\
  \hline
    $10^{-1}$ & $8.143 \phantom{\cdot 10^{-0}}$ & $2.981 \cdot 10^{-1}$ & $1.068 \cdot 10^{2}$ & $4.746 \phantom{\cdot 10^{-0}}$ \\
    $10^{-2}$ & $6.434 \cdot 10^{-1}$ & $1.395 \cdot 10^{-3}$ & $5.728 \cdot 10^{1}$ & $1.679 \phantom{\cdot 10^{-0}}$ \\
    $10^{-3}$ & $4.119 \cdot 10^{-1}$ & $1.425 \cdot 10^{-5}$ & $5.782 \cdot 10^{-1}$ & $8.385 \cdot 10^{-2}$ \\
    $10^{-4}$ & $3.969 \cdot 10^{-1}$ & $1.427 \cdot 10^{-7}$ & $4.097 \cdot 10^{-1}$ & $7.681 \cdot 10^{-3}$ \\
\end{tabular}

  \vspace{0.25cm}

  \begin{tabular}{c|cc|cc}
  & \multicolumn{2}{c|}{Measure constraint} & \multicolumn{2}{|c}{Velocity constraint} \\
  $h$ & $\energyError^M$ & $\constraintError^M$ & $\energyError^M$ & $\constraintError^M$ \\
  \hline
    $2^{-4}$ & $3.250 \phantom{\cdot 10^{-0}}$ & $1.164 \cdot 10^{-7}$ & $3.352 \phantom{\cdot 10^{-0}}$ & $5.299 \cdot 10^{-3}$ \\
    $2^{-5}$ & $1.063 \phantom{\cdot 10^{-0}}$ & $1.220 \cdot 10^{-7}$ & $1.100 \phantom{\cdot 10^{-0}}$ & $6.469 \cdot 10^{-3}$ \\
    $2^{-6}$ & $5.473 \cdot 10^{-1}$ & $1.278 \cdot 10^{-7}$ & $5.662 \cdot 10^{-1}$ & $6.679 \cdot 10^{-3}$ \\
    $2^{-7}$ & $3.969 \cdot 10^{-1}$ & $1.427 \cdot 10^{-7}$ & $4.097 \cdot 10^{-1}$ & $7.681 \cdot 10^{-3}$ \\
\end{tabular}

  \caption{Convergence results for our test problem. Errors are given at the final time step.
    The first table varies $\deltaT$ whilst keeping $h = 2^{-7}$ fixed.
    The second table varies $h$ whilst keeping $\deltaT = 10^{-4}$ fixed.}
  \label{tab:conv}
\end{table}

\begin{figure}[tbh]
  \centering
  \includegraphics{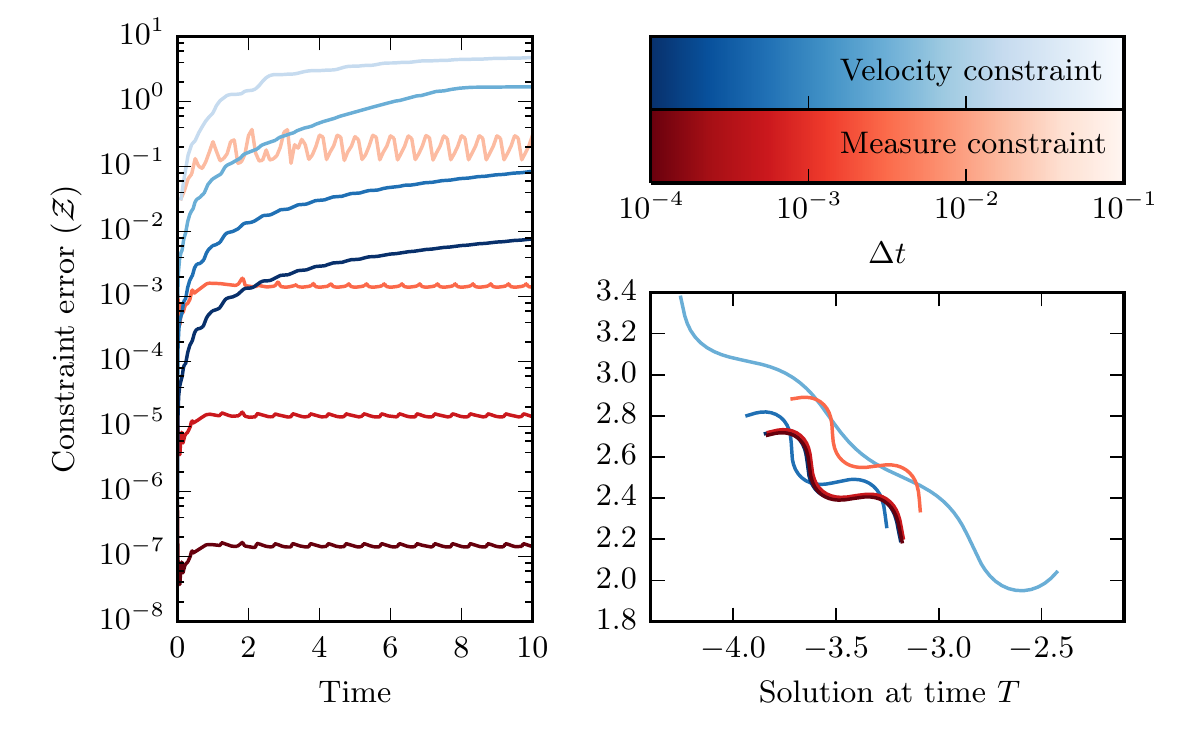}
  \caption{Convergence results in $\Delta T$ for a test problem. Left shows the error in the constraint over time $(\constraintError^k)$ and right shows the solution $\vec{x}_h^M$ at the final time. The colours are indicated by the colour bar at the top right. Red plots correspond to the measure constraint and blue plots correspond to the velocity constraint.}
  \label{fig:conv}
\end{figure}

\subsubsection*{The model replicates expected forward locomotion behaviours
over a wide range of conditions and parameter values}

To explore the dependence of the model on parameters, we ask, first, whether
model implementation with default parameters from the literature
(Table~\ref{tab:parameters}) can replicate reported locomotion behaviours.
To simulate forward locomotion, we set the muscle activation function
$\beta$ to be given in dimensional form by a travelling sine-wave with a
decaying amplitude
\begin{linenomath*}
\begin{equation}
  \label{eq:travelingwave}
  \begin{aligned}
  \beta( u, t ) & = \beta_0(u) \sin( 2 \pi u L / \lambda - 2 \pi \omega t ), \\
  \beta_0(u) & = \big( 10 (1-u) + 6 u \big) \unit{mm}^{-1}, \, \lambda = 0.65 \unit{mm}, \, \omega = 0.30 \unit{s}^{-1}.
\end{aligned}
\end{equation}
\end{linenomath*}
where $u$ follows a material point along the worm (varying from 0 at the head
to 1 in the tail).
The choice of $\beta_0$ reflects the  maximum curvatures exhibited by a worm during forwards locomotion, see \citet[Fig. 3, agar]{BoyBerCoh12}, and typical wavelength from during a crawling gate \citep{BerBoyTas09,FanWyaXie10}.

\begin{figure}[htb]
  \centering
  \includegraphics{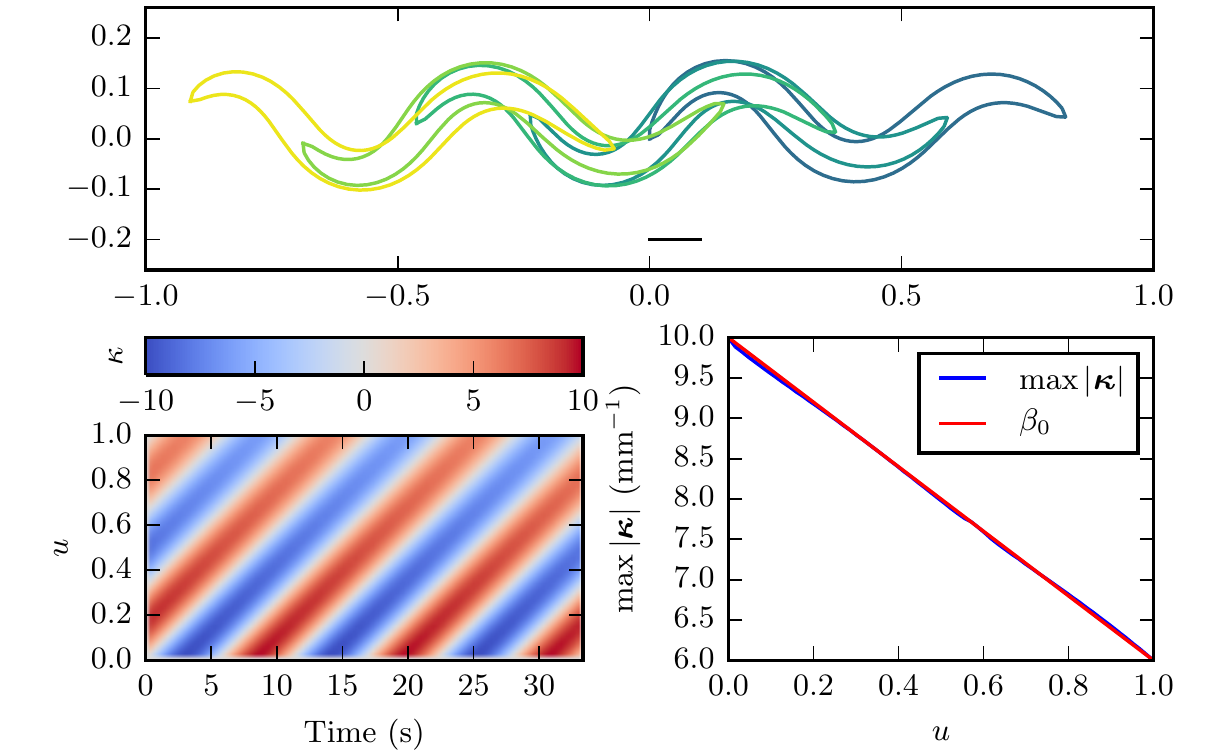}
  \caption{Simulation results for forced forward locomotion. The top panel shows simulated body postures in time, in response to crawling like muscle forcing at times $0, 5, 10, 15 \unit{s}$ (from blue to yellow). The initial condition is fixed to remove any transient behaviour.
    The axis are in mm with a $0.1 \unit{mm}$ scale bar shown.
    The bottom left panel shows the travelling waves of curvature along the body. The bottom right panel shows the maximum curvature at each point along the body (after the initial transient), as compared to the muscle forcing function. Recall that $u = 0$ denotes the head and $u = 1$ the end of the tail.}
  \label{fig:typical}
\end{figure}

Figure~\ref{fig:typical} features both typical body postures and a plot showing
the curvature along the body for agar like external fluid forces. The postures
show a near-sinosoidal configuration with curvature travelling from head to
tail down the body.  We see that the maximum curvature along the body follows
the desired curvature imposed by the muscle forcing. In other words, as
expected for agar like conditions and sufficiently high body elasticity, the
passive body properties enable the body curvature to closely follow the muscle
activation.  Note in particular, the negligible transient time before stable
undulations are observed, suggesting that the model would perform robustly in
dynamics environments or under various modulations of internal or external
forcing. Thus the model is able to replicate experimentally observed locomotion
body postures for realistic choices of parameter values.  

Next to validate the model of the linear viscoelastic environment, we test
whether the thrust generated by the worm during forward locomotion agrees with
reported values on agar. Simulations under agar like
conditions (with a ratio of drag coefficients $K=39$) generate an absolute
speed of $0.16 \unit{mm}/\unit{s}$ which compares well with previously observed speeds ($0.137 \pm 0.014 \unit{mm}/\unit{s}$, \citealt{Yemini2013}).

\subsubsection*{The model performs well over a wide range parameters and
body shapes}

\begin{figure}[htb]
  \centering
  \includegraphics{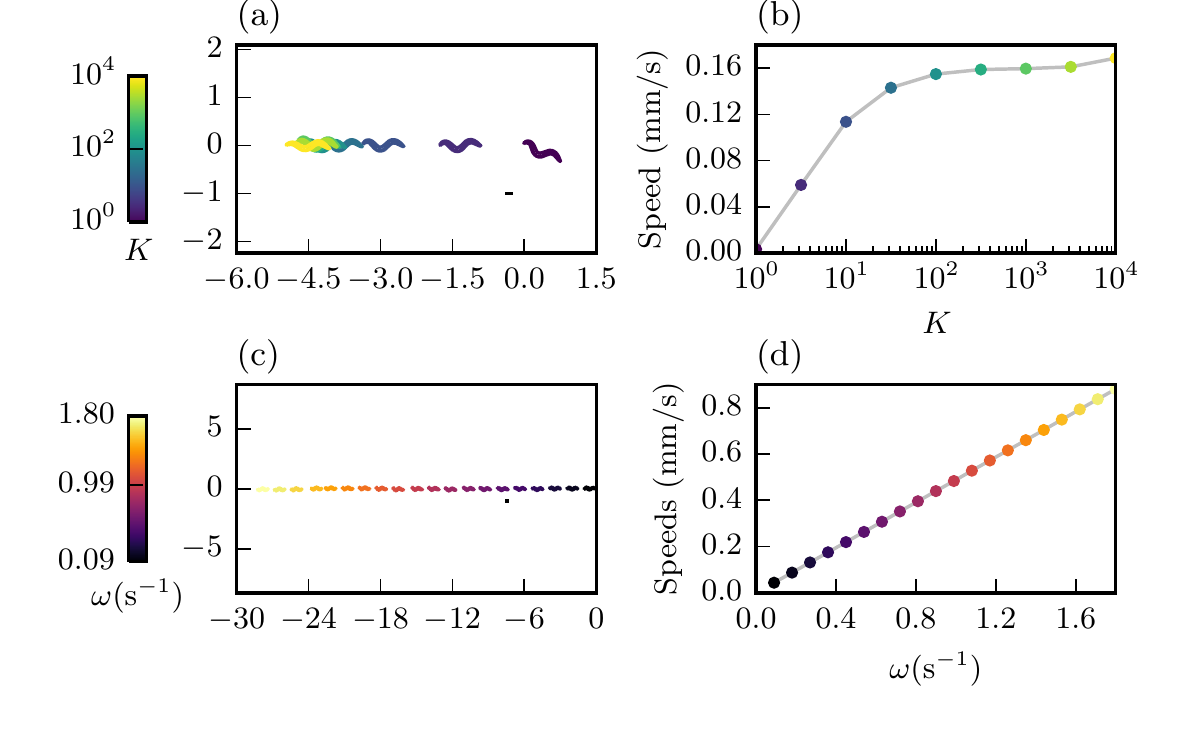}
  \caption{Modulation of the locomotion by internal and external parameters:
  $K$ (a,b), $\omega$ (c,d).
  The left column of plots (a,c) show the final shape after $28.3 \unit{s}$ starting with the head at $(0,0)$.
% $\vec{x}_0(u) = (0,u) \unit{mm}$.
  All axes are in $\unit{mm}$ with a $0.1 \unit{mm}$ scale bar in the bottom right of each plot.
  The right column of plots (b,d) show the average speed during the trial (ignoring undulations).}
  \label{fig:varied}
\end{figure}

\begin{figure}[htb]
  \centering
  \includegraphics{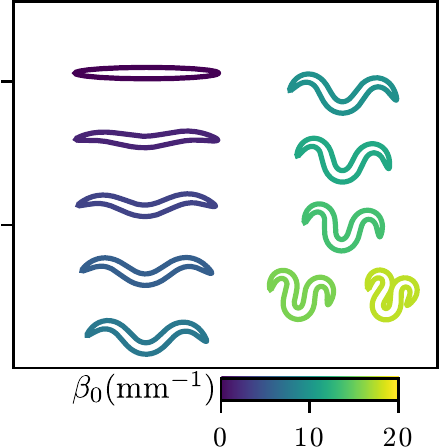}
  \caption{Simulation results for different strengths of muscle forcing 
$\beta_0$. All plots show the final shape after $33.3 \unit{s}$ starting 
from a straight line configuration, and subject to constant forcing $\beta$.  
The shapes have been moved on the plot to demonstrate the different postures.}
  \label{fig:varied-beta}
\end{figure}

\textit{C.\ elegans} can independently modulate its undulation frequency and
waveform, within some limits. Furthermore, gait modulation (correlated across
frequency and waveform) is observed as a function of the viscosity (or
viscoelasticity) of the fluid environment. In our model, we are able to vary
these properties independently in order to test the range of behaviours the
model can generate.
% Worm undulation periods vary from approximately 0.5 second (rapid, long-wavelength thrashing or swimming in liquid) up to 3-4 seconds (slow, short wavelength crawling on agar), with slower undulations in more resistive environments \citep{BerBoyTas09,FanWyaXie10}.
Figure~\ref{fig:varied}(c,d)
shows worms simulated with the same forcing function, but variable undulation frequencies from $0.1 \unit{s}^{-1}$ to $1.8 \unit{s}^{-1}$. This covers the range of parameters from Table~\ref{tab:parameters}.
The model copes well with the full
range of undulation frequencies, with the body curvature following the forcing for all.
In other words, different values of $\omega$ lead to different resulting speeds but with similar body postures.
By contrast,
changing the wavelength of undulations $\lambda$, leads to a range of different 
postures as well as speeds (see Figure~\ref{fig:varied-lambda}(a)).
As expected, for a given waveform, for low frequencies the faster the 
undulation, the faster the speed.
Specifically, for low values of $\omega$, we find a linear relationship between undulation frequency and forward speed, indicating that the undulation frequency merely scales the effective time of the simulation.

Next we use the model to validate behaviours across a range of
external viscoelasticities by varying the ratio of drag coefficients $K$, while keeping
the undulation waveform and frequency fixed.
% For sufficiently shallow waveforms, \citet{GraLis64} showed that this ratio dictates the forward progress per undulation (or conversely, the percentage `slip' of the wavelength relative to the surface or medium).  Thus $K=1$ corresponds to an isotropic environment, where no thrust is generated, Newtonian media with $K\approx 1.5$ are characterised by high `slip', as compared to agar, and the limit $K\rightarrow\infty$ corresponds to a solid channel (with no slip).
We simulated forward locomotion for $K$ ranging from 1 to 10,000. Specifically, we vary $K_{\vec\nu}$ whilst keeping $K_{\vec\tau}$ fixed. Thus an increase in $K$ corresponds to a stiffer medium.
We see a lack of thrust for $K=1$, and maximum thrust for the largest values of $K$. We
further find that for $K\approx 100$, the distance travelled saturates, thus
increasing $K$ has a negligible effect on speed. Note
for this choice of $E$, varying the ratio of environmental parameters $K \lesssim 100$ does
not significantly affect the body postures in our simulations in this parameter range.
This corresponds to the theory of \citet{GraLis64} who that for sufficiently small undulations that the ratio $K$ dictates the forward progress per undulation (or conversely the percentage `slip' relative to the background media). For $K=1$, we see total slip and for $K \sim 100$, we see minimal slip.
For $K \gtrsim 1000$, the environment is too stiff for the worms muscles to overcome and less thrust is produced.
Note that we expect the worm not to be able to locomote in $K \approx 1000$, which suggests that our value $E = 10 \unit{MPa}$ is too high (see below).

\begin{comment} This is expected since we do not modulate the muscle 
forcing parameters with $K$. One would expect an integrated model to show this variation.
We see that for very low values of $K$ there is very little distance travelled and for higher values, $K \gtrsim 100$, the distance travelled saturates.
\end{comment}

Finally, we probe the flexibility of our method in generating arbitrary body
shapes.  To do this, we simulate our model with similar muscle forcing function
(\ref{eq:travelingwave}) but setting the amplitude of preferred undulations
$\beta_0(u) = \beta_0$ to constant values from $0$-$20 \unit{mm}^{-1}$. The
results are shown in Figure~\ref{fig:varied-beta}.  We see that as we increase
$\beta_0$ the posture is increasingly curved.  Note that for the highest value
of $\beta_0$, the body of the worm self intersects which is not precluded from
our assumptions on the fluid dynamics.  This shows the highest curvatures that
the model can capture.
% The highest $\beta_0$ corresponds to a muscle contraction of 40\%, which is well beyond that of physiological muscles.
The ability to arbitrarily and reliably modulate the driving force suggests a powerful tool
to study a variety of behaviours (for example, omega turns) as well as
biophysical properties of muscle and cuticle over development in wild type
animals, as well in animal models of disease.

\subsection{Insight on nematode locomotion}
\subsubsection*{The model places constraints on the range of viable material properties of the worm}

\begin{figure}[tbh]
  \centering
  \includegraphics{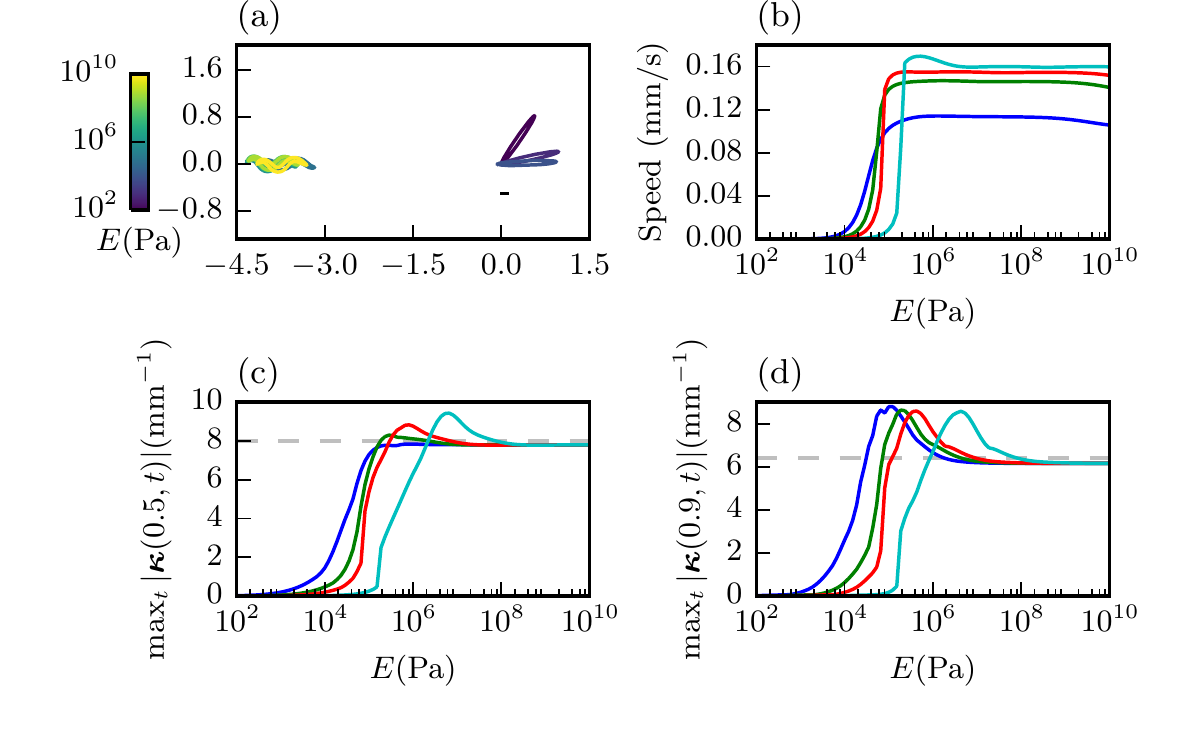}
  \caption{The effect of Young's modulus ($E$) on locomotion.
    (a) The final body shape after $28.3 \unit{s}$ starting with the head at $(0,0)$. Only $E$ varies. All other parameters fixed.
    (b) The average speed during the trial (ignoring undulations) for varying $E$ for different values of $K$ (Blue $K=10$, green $K=39$, red $K=100$, cyan $K=1000$).
    % (c-f) Maximum curvature achieved along the body for different values of $K$. The colour scheme is the same as (a).
    (c,d) Maximum curvature achieved at the centre of the body ($u=0.5$) and near the tail ($u=0.9$) for different values of $K$ (the same colours as (b)). The dashed grey line represents $\beta_0(u)$.}
  \label{fig:varied-E}
\end{figure}

A key question is whether the key free parameter of the model, the Young's
modulus $E$, may be constrained by the model.  For the above simulations, we
chose a Young's modulus of $E=10 \unit{MPa}$.
%While well within the range of
%experimentally reported values, we were surprised that the worm was
%sufficiently elastic to successfully undulate even at unrealistically high
%values of $K$.
However the range of reported values ranges over five orders of magnitude.
We therefore asked how the ability of muscle forcing to impose
body curvature depends on internal elasticity.  To answer this question, we
repeated the above simulations varying $E$ from $1 \unit{kPa}$ to $10 \unit{GPa}$ with $K$ from $10$ to $1000$ (see Figure~\ref{fig:varied-E}).

For low values of the Young's modulus, $E$, the worm is not sufficiently stiff to overcome the external environmental stiffness which results in a large error between the curvature of the worm and the preferred curvature. The lack of generated curvature results in almost no forwards thrust.
In contrast, for high values of $E$ the desired curvature is exactly achieved which leads to forwards locomotion.
We note that for intermediate values of $E$, the curvature is higher than the muscle activation profile ($\beta_0$), which leads to even higher speeds.

The values in Table~\ref{tab:parameters} list $K = 39$ for agar. For $K = 10$-$39$, we require that  the worm's
internal elasticity suffices to generate high thrust and high curvatures. For these values, we see the transition from low to high thrust and high to low curvature difference around $E \approx 10 \unit{kPa}$ (Figures 6(b,c,d) green line).
We explore $K = 1000$, which corresponds to $(K_{\vec\nu}, K_{\vec\tau}) = (3200, 3.2) \cdot \unit{kg} \unit{m}^{-1} \unit{s}^{-1}$ and much larger drag than $10 \unit{Pa} \unit{s}$, thus we expect to see low locomotion speeds here and for the worm to struggle to generate curvature.
For $K = 1000$, we see the transition from low to high thrust and high to low curvature error around $E \approx 1 \unit{MPa}$ (Figures~\ref{fig:varied-E}(b,c,d) cyan line). Combining these conservative bounds we give an estimate of $10 \unit{kPa} \lesssim E \lesssim 1 \unit{MPa}$.

\begin{comment}
The value we use for $E$ ($10 \unit{MPa}$) is within the upper range which is less sensitive to parameter variations in our current setting of high external fluid viscosity.
%
The sensitivity of the worm's locomotion to the parameter choices is shown in Figure~\ref{fig:varied}. 
%
For our choice of $K$ and $E$, scaling the frequency of undulations, $\omega$, simply rescales the effective time. Different values of $\omega$ lead to different resulting speeds but with similar body postures.

{\color{green}
%\notes{as you make a big deal about ranges of speed and realistic speeds, it makes sense to give some values here. What does the speed saturate to for large $K$ for example and below for different postures?}
%\notes{TR: I don't think I am making that big a deal of this - only that it is possible to say something. here the speed is $0.05 \unit{mm}/\unit{s}$ which is probably too slow but relates to a beat frequency of $0.3 \unit{s}^{-1}$. This is the value stated by Fang-Yang but I would say this is perhaps too slow also. recall speed scales with beat frequency.\ldots}
}
\end{comment}

\subsubsection*{The model suggests that the worm modulates its waveform to maintain speed}

\begin{figure}[tb]
  \centering
  \includegraphics{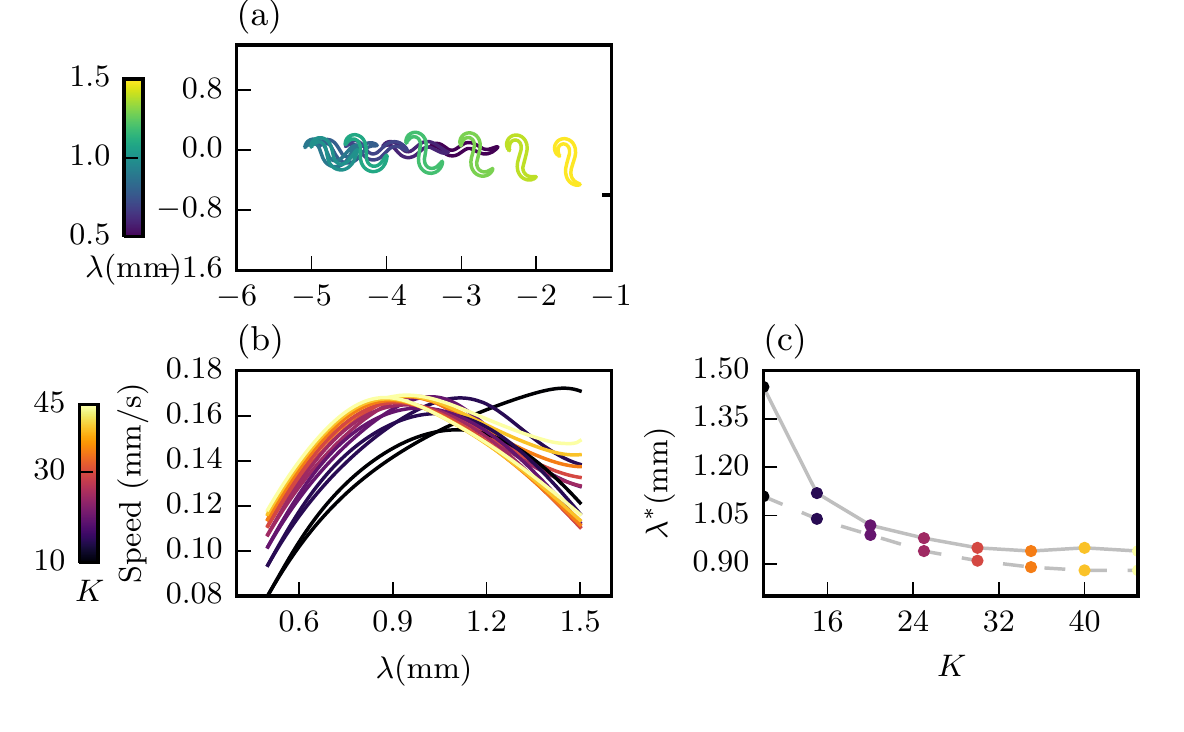}
  \caption{The undulation wavelength modulates locomotion speed.
    (a) The final body shape after $28.3 \unit{s}$ starting with the head at $(0,0)$. Only $\lambda$ varies. All other parameters fixed.
    % (b) Maximum curvature achieved along the body. The colour scheme is the same as (a).
    (b) The average speed during the trial (ignoring undulations) for varying $\lambda$ for different values of $K$.
    (c) The optimal choice of $\lambda = \lambda^*$ for each $K$. The solid line represents $E = 10 \unit{MPa}$, the dashed line represents $E = 1 \unit{MPa}$.
}
  \label{fig:varied-lambda}
\end{figure}

Finally we asked how wavelength modulation affects locomotion speed in different viscoelastic media.
Changing the wavelength of undulations $\lambda$ (Figure~\ref{fig:varied-lambda}) shows a wide range of body postures and swimming speeds.
\citet{BerBoyTas09} reported a smooth decrease of the wavelength, amplitude and frequency of undulations as the worm locomotes through increasingly resistive media,
corresponding to increasing $K$.
Similar forms of smooth gait modulation spanning the swim-crawl transition have also been observed for Newtonian and near-Newtonian media \citep{FanWyaXie10,lauga,Lebois2012}.
We find that different choices of $K$ correspond to different optimal wavelengths. For agar like environments, our model yields a maximum speed for a wavelength of $\lambda_{\mathrm{opt}} (K = 39) = 0.9 \unit{mm}$.
The optimal wavelength (higher than the wavelength reported under similar conditions) appears relatively insensitive to the choice of $E$ (at $1 \unit{MPa}$ and up), but depends more strongly on the choice of waveform.
As $K$ is reduced, we find a that the optimal wavelength increases
significantly to $\lambda_{opt} (K = 10) = 1.5 \unit{mm}$, commensurate with a swimming wavelength
(Table~\ref{tab:parameters}). As expected, in these low resistance environments, the optimal
wavelength depends strongly on the material properties of the worm.
The qualitative correspondence of the above results with the observed waveform modulation
in the swim-crawl transition suggests that the worm must modulate its undulation wavelength to maximise speed across a range of environments.
The simulations are repeated at $E = 1 \unit{MPa}$ to show qualitatively the same results, and quantitatively the similar results for sufficiently resistive environments, $K \ge 20$.

\section{Discussion}
\label{sec:conclusion}

\begin{comment}
{\color{magenta}
\noindent
$\bullet~~~$ 
Added: non-arclength. peristalsis.
\newline $\bullet~~~$ closing assumption. Is this a problem for
internal viscosity?
\newline $\bullet~~~$ 
To add: improvement in efficiency over articulated body models (boyle 2012)
and particle systems simulations (openworm). Something added...
\newline $\bullet~~~$ 
improvement in numerical stability over articulated body models
(where numerical instabilities occurred, e.g. for some combinations
of drag coeff.)
\newline $\bullet~~~$ 
Can you comment on what it would take to extend the numerical method 
to nonzero viscosity? By adding viscosity would a new closing
assumption be needed for the system not to be underdetermined?
\newline $\bullet~~$
selected results: added lower bound on $E$, what about upper bound (failure
of undulations--in reasonable time--for realistic values of $K$?
can you say something about relaxation times in different media?
\newline $\bullet~~~$
selected results: are both K and fluid viscosity modulations possible?
Here we chose to only show modulations of $K$. Is there a problem
with just increasing viscosity and keeping $K$ value Newtonian (1.5)?
%
added parameter ranges and how do they compare to experimental findings?
}\\
\end{comment}

In all animals, the shape of the organism is determined by integrating the 
effects of internal driving forces (muscle actuation) and the interaction of 
the passive body with its external environment. The combined action of active 
and passive moments in determining body shape is particularly important in 
the study of active swimming and undulatory locomotion, where spatio-temporal 
feedback about the shape of the body can stabilise the motion, facilitate
adaptation and modulation of posture, and -- in \textit{C.\ elegans} -- 
even generate the rhythmic patterns of neuromuscular activity underlying 
locomotion. An integrated study of locomotion in any such system requires a
flexible, stable and efficient model of the biomechanics to interface
to model nervous systems. To date, no such model exists.

To address this need, we have detailed a biomechanical model for 
undulatory locomotion of slender bodies, applied to a case study of 
\textit{C.\ elegans}.  The model
incorporates the collective passive viscoelasticity of the tissue, an active
moment, and drag forces capturing the interaction of the organism with the
environment at low Reynolds numbers.  In this model, the body of the worm is
represented by a tapered elongated shell, driven by periodic travelling waves
of muscle contraction. Effective force equations on the midline of the body are
derived. While the force equations have been adapted from the literature
\citep{LanLif75,GuoMad08}, in themselves they yield an underdetermined system
of equations. To close the system, we further consider conservation of mass which leads to a local length
constraint. Our continuum description and closing assumption leads to a
nonlinear initial-boundary value problem, which we solve for an anatomically meaningful
parametrisation. Parameter values for the geometry of the
body, passive and active moments, and the environmental forces are taken from
anatomical and behavioural data from the literature.

Our finite element approach allows us to directly solve the non-linear system
of equations and in this way differs from previous continuous models of
\textit{C.\ elegans} biomechanics~\citep{SznPurKra10,FanWyaXie10}.
This approach ensures the high accuracy and stability of the model and permits
the simulations of arbitrarily large (even non-physical) curvatures.
To this end, we have presented a new numerical method based on a mixed
piecewise-linear piecewise-constant finite element method and semi-implicit
time stepping. We have shown the robustness of the scheme analytically through a semi-discrete stability bound and good
equi-distribution properties for the fully discrete scheme.
Numerical experiments show the accuracy of the method.
Our simulations demonstrate that the numerical methodology is robust with respect to a large variety of body shapes and has improved stability over previous articulated models \citep{Ding2012,BoyBerCoh12} (see in particular \citet[Fig.\ 4C]{BoyBerCoh12} which shows a large, behaviourally significant region of parameter space is unachievable with the authors' approach).
To our knowledge, the only other formulation of a \textit{C.\ elegans} model with the potential to generate complex waveforms 
and high curvature postures relies on particle system simulations 
\citep{Palyanov2016}.
It is our hope that the model presented here offers similar flexibility
and generalisability, but with significantly lower computational overheads.

We contrast our approach with two previous approaches to computational modelling of \textit{C.\ elegans} locomotion. First, \citet{BoyBerCoh12} used a discrete
models in which an articulated body is modelled as a repeating pattern of
springs and dampers.
The continuous nature of the elastic beam is
particularly appealing since it more closely mimics the continuous,
non-segmented body of nematodes and provides a more robust simulation framework.
Second, the OpenWorm consortium \citep{Szigeti2014} has built a cell-scale reconstruction of the worm and uses smoothed particle hydrodynamics (SPH) in order to capture physical forces \citep{Palyanov2016}.
By comparison to the SPH approach, this work assumes simplified fluid dynamics and uniform density; this allows us to reduce the required number of degrees of freedom by four orders of magnitude which in turn yields a more efficient method.

The stability of the model with respect to exogenous (fluid) and endogenous
(body and forcing) parameters and its ability to reliably generate dynamics
with high curvature postures make it suitable for further studies of the
biomechanics of this and other soft bodied systems. In particular, the choice
of parametrisation provides a general framework for following material 
coordinates along a body (as distinct from the conventional arc-length
parametrisations). In \textit{C.\ elegans} and other
side-to-side undulators in which the shape of the body is conserved locally,
our formulation allows us to naturally impose the forces resulting from
internal pressure. In larvae and other peristaltic undulators, distances are
not locally conserved along the body, and the application of this formulation
may therefore offer a powerful numerical framework for modelling similar 
systems.

In \textit{C.\ elegans}, the model is particularly well suited to the study of
a range of biological behaviours (e.g., locomotion, coiling and turning) and
biological variants (e.g.,  disease models, genetic knock-out animals,
developing or ageing animals) in a realistic biomechanical framework and across
a wide range of Newtonian and viscoealstic physical environments that may be
approximately described by drag forces. Potential examples of applications for
\textit{C.\ elegans}, where current simulators would fail include coilers,
loopy worms, omega turns and more.

The model allows us, not only to generate a large range of body shapes, but 
also to probe 
relevant exogeneous and endogenous parameters of the system. 
%Simulations presented here confirm previous estimates of resistive forces 
% on agar ($K_{\nu}/K_{\tau}\approx 40$) \citep{BerTasBoy09}.
%Furthermore, 
In particular, our model allows us to address a long-standing open
question in \textit{C.\ elegans}: estimating the effective elasticity of 
the body. Our model places a estimates on the elasticity of the tissue 
(Young's modulus $E\gtrsim 10\unit{kPa}, E \lesssim 1 \unit{MPa}$), required to overcome realistic 
resistive forces and generate forward thrust.
These results are consistent the estimate of \citet{BacRyuDal13} with prior qualitative experimental observations but higher than the estimate of \citet{SznPurKra10} and lower than the estimates of \citet{ParGooPru07,FanWyaXie10}.

If the Young's modulus is too low, this would restrict the behavioural repertoire
of the worm. In particular, we are able to set a lower bound on the Young's modulus
by requiring that the worm can both bend and generate thrust in agar-like conditions. 
Conversely, we set an upper bound based on behavioural observations that
the worm fails to generate thrust in sufficiently resistive environments. Intuitively, 
the stiffer the cuticle, the higher the power-output of the muscles required to achieve 
the desired bending. Furthermore, for agar like environments, the locomotion of the 
worm is not power limited \citep{BerBoyTas09}. However, resistive force theory suggests
that the worm must considerably increase its muscle power to navigate
through increasingly resistive media so that at sufficiently high viscosity or
viscoelasticity, power limitations dominate. Thus, it is not known whether a higher Young's
modulus may incur a prohibitive material cost, a prohibitive muscle power, or
whether such high viscoelasticity environments are sufficiently rare in the worm's
natural habitat to preclude such a need. Finally, we show that for a given choice
of Young's modulus, waveform modulation allows worms to optimise speed as they
maneuvre through different viscous and viscoelastic environments.

Combined, the generality of the model and combination of analysis and 
numerical simulations that demonstrate the good properties of our method, 
suggest broad applicability for the simulation of locomotory behaviours in 
non-segmented soft-bodied systems and predictions of not only body shapes but also swimming speeds.
Future extensions of the numerical scheme include effectively capturing internal viscosity. This will allow the model to also be used in lower resistance environments and, for example, capture in more detail the transition from crawling to swimming.

This work focuses on the biomechanical aspects of locomotion and stops short of
integrating the dynamics of the nervous system which is responsible for the rhythmic control that activates the muscles, and thus generates the
locomotion. Instead, we use prescribed muscle activation functions similarly to
\cite{SznPurKra10} to actuate the passive body.
Our formulation of the model lays the ground for a robust and
efficient integration of closed-loop control including nervous system, body and
environment. The efficiency and stability displayed by this numerical approach suggest that an appropriate integrated approach, including neuromuscular control and proprioceptive feedback could make predictions about non-undulatory behaviours such as complex navigation and motion through granular or heterogeneous environments.

%
%

% only crawling on agar
% no internal viscosity
% improved stability
% more flexibility wrt body position
% easy modulation of important physical parameters
% integrated neuromuscular control left for future work.
% the efficiency and stability displayed suggests that integrated model 
% will allow extensions to non-undulatory behaviours (complex navigation, 
% omega turns, etc.).

\section*{Acknowledgements}

The research was funded by the Engineering and Physical Sciences Research Council (EPSRC EP/J004057/1).
The authors wish to thank Peter Bollada, Samuel Braustein, Robert Holbrook and Ian Hope for interesting discussions and insights which have helped in the preparation of this manuscript.

% \section*{References}

\bibliographystyle{elsarticle-harv}
\bibliography{library}

\begin{thebibliography}{60}
\expandafter\ifx\csname natexlab\endcsname\relax\def\natexlab#1{#1}\fi
\expandafter\ifx\csname url\endcsname\relax
  \def\url#1{\texttt{#1}}\fi
\expandafter\ifx\csname urlprefix\endcsname\relax\def\urlprefix{URL }\fi

\bibitem[{Backholm et~al.(2015{\natexlab{a}})Backholm, Kasper, Schulman, Ryu,
  and Dalnoki-Veress}]{Backholm2015a}
Backholm, M., Kasper, A. K.~S., Schulman, R.~D., Ryu, W.~S., Dalnoki-Veress,
  K., sep 2015{\natexlab{a}}. The effects of viscosity on the undulatory
  swimming dynamics of \textit{C.\ elegans}. Physics of Fluids 27~(9), 091901.
\newline\urlprefix\url{http://dx.doi.org/10.1063/1.4931795}

\bibitem[{Backholm et~al.(2013)Backholm, Ryu, and Dalnoki-Veress}]{BacRyuDal13}
Backholm, M., Ryu, W.~S., Dalnoki-Veress, K., 2013. Viscoelastic properties of
  the nematode caenorhabditis elegans, a self-similar, shear-thinning worm.
  Proceedings of the National Academy of Sciences 110~(12), 4528--4533.
\newline\urlprefix\url{http://www.pnas.org/content/110/12/4528.abstract}

\bibitem[{Backholm et~al.(2015{\natexlab{b}})Backholm, Ryu, and
  Dalnoki-Veress}]{Backholm2015}
Backholm, M., Ryu, W.~S., Dalnoki-Veress, K., may 2015{\natexlab{b}}. The
  nematode \textit{C.~elegans} as a complex viscoelastic fluid. The European
  Physical Journal E 38~(5).
\newline\urlprefix\url{http://dx.doi.org/10.1140/epje/i2015-15036-1}

\bibitem[{Barrett et~al.(2010)Barrett, Garcke, and N\"{u}rnberg}]{BarGarNur10}
Barrett, J.~W., Garcke, H., N\"{u}rnberg, R., sep 2010. The approximation of
  planar curve evolutions by stable fully implicit finite element schemes that
  equidistribute. Numerical Methods for Partial Differential Equations 27~(1),
  1--30.
\newline\urlprefix\url{http://dx.doi.org/10.1002/num.20637}

\bibitem[{Barrett et~al.(2011)Barrett, Garcke, and N\"{u}rnberg}]{BarGarNur11}
Barrett, J.~W., Garcke, H., N\"{u}rnberg, R., oct 2011. Parametric
  approximation of isotropic and anisotropic elastic flow for closed and open
  curves. Numerische Mathematik 120~(3), 489--542.
\newline\urlprefix\url{http://dx.doi.org/10.1007/s00211-011-0416-x}

\bibitem[{Bartels(2013)}]{Bar13}
Bartels, S., jan 2013. A simple scheme for the approximation of the elastic
  flow of inextensible curves. {IMA} Journal of Numerical Analysis 33~(4),
  1115--1125.
\newline\urlprefix\url{http://dx.doi.org/10.1093/imanum/drs041}

\bibitem[{Batchelor(2000)}]{Batchelor2000}
Batchelor, G.~K., 2000. An Introduction to Fluid Dynamics. Cambridge
  Mathematical Library.

\bibitem[{Berri et~al.(2009)Berri, Boyle, Tassieri, Hope, and
  Cohen}]{BerBoyTas09}
Berri, S., Boyle, J.~H., Tassieri, M., Hope, I.~A., Cohen, N., 2009. {Forward
  locomotion of the nematode \textit{C. elegans} is achieved through modulation
  of a single gait}. HFSP Journal 3~(3), 186--193, pMID: 19639043.
\newline\urlprefix\url{http://dx.doi.org/10.2976/1.3082260}

\bibitem[{Boyle et~al.(2012)Boyle, Berri, and Cohen}]{BoyBerCoh12}
Boyle, J.~H., Berri, S., Cohen, N., 2012. Gait modulation in \textit{C.\
  elegans}: An integrated neuromechanical model. Front. Comput. Neurosci. 6.
\newline\urlprefix\url{http://dx.doi.org/10.3389/fncom.2012.00010}

\bibitem[{Butler et~al.(2014)Butler, Branicky, Yemini, Liewald, Gottschalk,
  Kerr, Chklovskii, and Schafer}]{Butler2014}
Butler, V.~J., Branicky, R., Yemini, E., Liewald, J.~F., Gottschalk, A., Kerr,
  R.~A., Chklovskii, D.~B., Schafer, W.~R., nov 2014. A consistent muscle
  activation strategy underlies crawling and swimming in caenorhabditis
  elegans. Journal of The Royal Society Interface 12~(102), 20140963--20140963.
\newline\urlprefix\url{http://dx.doi.org/10.1098/rsif.2014.0963}

\bibitem[{Cohen and Boyle(2010)}]{CohBoy10}
Cohen, N., Boyle, J.~H., Mar 2010. Swimming at low reynolds number: a beginners
  guide to undulatory locomotion. Contemporary Physics 51~(2), 103–123.
\newline\urlprefix\url{http://dx.doi.org/10.1080/00107510903268381}

\bibitem[{Cohen and Sanders(2014)}]{Cohen2014}
Cohen, N., Sanders, T., apr 2014. Nematode locomotion: dissecting the
  neuronal{\textendash}environmental loop. Current Opinion in Neurobiology 25,
  99--106.
\newline\urlprefix\url{https://doi.org/10.1016%2Fj.conb.2013.12.003}

\bibitem[{Coomer et~al.(2001)Coomer, Lazarus, Tucker, Kershaw, and
  Tegman}]{COOMER_2001}
Coomer, J., Lazarus, M., Tucker, R., Kershaw, D., Tegman, A., Feb 2001. A
  non-linear eigenvalue problem associated with inextensible whirling strings.
  Journal of Sound and Vibration 239~(5), 969–982.
\newline\urlprefix\url{http://dx.doi.org/10.1006/jsvi.2000.3190}

\bibitem[{Corsi(2015)}]{wormbook}
Corsi, A.~K., jun 2015. A transparent window into biology: A primer on
  caenorhabditis elegans. {WormBook}, 1--31.
\newline\urlprefix\url{http://dx.doi.org/10.1895/wormbook.1.177.1}

\bibitem[{Croll(1970)}]{Cro70}
Croll, N.~A., 1970. The behaviour of nematodes: their activity, senses and
  responses. The behaviour of nematodes: their activity, senses and responses.

\bibitem[{Cronin et~al.(2005)Cronin, Mendel, Mukhtar, Kim, Stirbl, Bruck, and
  Sternberg}]{Cronin2005}
Cronin, C.~J., Mendel, J.~E., Mukhtar, S., Kim, Y.-M., Stirbl, R.~C., Bruck,
  J., Sternberg, P.~W., 2005. An automated system for measuring parameters of
  nematode sinusoidal movement. {BMC} Genet 6~(1), 5.
\newline\urlprefix\url{http://dx.doi.org/10.1186/1471-2156-6-5}

\bibitem[{Davis(2004)}]{umfpack}
Davis, T.~A., jun 2004. Algorithm 832. {ACM} Transactions on Mathematical
  Software 30~(2), 196--199.
\newline\urlprefix\url{https://doi.org/10.1145%2F992200.992206}

\bibitem[{Deckelnick and Dziuk(2009)}]{DecDzi09}
Deckelnick, K., Dziuk, G., 2009. Error analysis for the elastic flow of
  parametrized curves. Mathematics of Computation 78~(266), 645--671.

\bibitem[{Deckelnick et~al.(2005)Deckelnick, Dziuk, and
  Elliott}]{Deckelnick2005}
Deckelnick, K., Dziuk, G., Elliott, C.~M., may 2005. Computation of geometric
  partial differential equations and mean curvature flow. Acta Numerica 14,
  139--232.
\newline\urlprefix\url{http://dx.doi.org/10.1017/S0962492904000224}

\bibitem[{Ding et~al.(2012)Ding, Sharpe, Masse, and Goldman}]{Ding2012}
Ding, Y., Sharpe, S.~S., Masse, A., Goldman, D.~I., dec 2012. Mechanics of
  undulatory swimming in a frictional fluid. {PLoS} Computational Biology
  8~(12), e1002810.
\newline\urlprefix\url{https://doi.org/10.1371%2Fjournal.pcbi.1002810}

\bibitem[{Dziuk(1990)}]{Dzi90}
Dziuk, G., dec 1990. An algorithm for evolutionary surfaces. Numerische
  Mathematik 58~(1), 603--611.
\newline\urlprefix\url{http://dx.doi.org/10.1007/BF01385643}

\bibitem[{Dziuk et~al.(2002)Dziuk, Kuwert, and Schatzle}]{DziKuwSch02}
Dziuk, G., Kuwert, E., Schatzle, R., jan 2002. Evolution of elastic curves in
  $\mathbb{R}^n$: Existence and computation. {SIAM} Journal on Mathematical
  Analysis 33~(5), 1228--1245.
\newline\urlprefix\url{http://dx.doi.org/10.1137/S0036141001383709}

\bibitem[{Elliott et~al.({1989})Elliott, French, and Milner}]{EllFreMil89}
Elliott, C.~M., French, D.~A., Milner, F.~A., {1989}. {A second order splitting
  method for the Cahn-Hilliard equation}. {Numerische Mathematik} {54}~({5}),
  575--590.

\bibitem[{Fang-Yen et~al.(2010)Fang-Yen, Wyart, Xie, Kawai, Kodger, Chen, Wen,
  and Samuel}]{FanWyaXie10}
Fang-Yen, C., Wyart, M., Xie, J., Kawai, R., Kodger, T., Chen, S., Wen, Q.,
  Samuel, A. D.~T., 2010. Biomechanical analysis of gait adaptation in the
  nematode caenorhabditis elegans. Proceedings of the National Academy of
  Sciences 107~(47), 20323--20328.
\newline\urlprefix\url{http://www.pnas.org/content/107/47/20323.abstract}

\bibitem[{Feng et~al.(2004)Feng, Cronin, Wittig, Sternberg, and
  Schafer}]{Feng2004}
Feng, Z., Cronin, C.~J., Wittig, J.~H., Sternberg, P.~W., Schafer, W.~R., 2004.
  An imaging system for standardized quantitative analysis of \textit{C.\
  elegans} behavior. {BMC} Bioinformatics 5~(1), 115.
\newline\urlprefix\url{http://dx.doi.org/10.1186/1471-2105-5-115}

\bibitem[{Fung(1993)}]{Fung1993}
Fung, Y.~C., 1993. {A First Course in Continuum Mechanics}. Prentice Hall,
  Englewood Cliffs, NJ.

\bibitem[{Gilpin et~al.(2015)Gilpin, Uppaluri, and Brangwynne}]{Gilpin2015}
Gilpin, W., Uppaluri, S., Brangwynne, C.~P., apr 2015. Worms under pressure:
  Bulk mechanical properties of c. elegans are independent of the cuticle.
  Biophysical Journal 108~(8), 1887--1898.
\newline\urlprefix\url{https://doi.org/10.1016%2Fj.bpj.2015.03.020}

\bibitem[{Gray and Hancock(1955)}]{GRAY802}
Gray, J., Hancock, G.~J., 1955. The propulsion of sea-urchin spermatozoa.
  Journal of Experimental Biology 32~(4), 802--814.
\newline\urlprefix\url{http://jeb.biologists.org/content/32/4/802}

\bibitem[{Gray and Lissmann(1964)}]{GraLis64}
Gray, J., Lissmann, H.~W., 1964. The locomotion of nematodes. Journal of
  Experimental Biology 41~(1), 135--154.

\bibitem[{Gray et~al.(2005)Gray, Hill, and Bargmann}]{Gray2005}
Gray, J.~M., Hill, J.~J., Bargmann, C.~I., feb 2005. A circuit for navigation
  in caenorhabditis elegans. Proceedings of the National Academy of Sciences
  102~(9), 3184--3191.
\newline\urlprefix\url{https://doi.org/10.1073%2Fpnas.0409009101}

\bibitem[{Guo and Mahadevan(2008)}]{GuoMad08}
Guo, Z.~V., Mahadevan, L., 2008. Limbless undulatory propulsion on land.
  Proceedings of the National Academy of Sciences 105~(9), 3179--3184.
\newline\urlprefix\url{http://www.pnas.org/content/105/9/3179.abstract}

\bibitem[{Hart(2006)}]{Hart2006}
Hart, A., 2006. Behavior. {WormBook}.
\newline\urlprefix\url{https://doi.org/10.1895%2Fwormbook.1.87.1}

\bibitem[{Helfrich(1973)}]{Helfrich1973}
Helfrich, W., jan 1973. Elastic properties of lipid bilayers: Theory and
  possible experiments. Zeitschrift f\"{u}r Naturforschung C 28~(11-12).
\newline\urlprefix\url{http://dx.doi.org/10.1515/znc-1973-11-1209}

\bibitem[{Karbowski et~al.(2006)Karbowski, Cronin, Seah, Mendel, Cleary, and
  Sternberg}]{Karbowski2006}
Karbowski, J., Cronin, C.~J., Seah, A., Mendel, J.~E., Cleary, D., Sternberg,
  P.~W., oct 2006. Conservation rules, their breakdown, and optimality in
  caenorhabditis sinusoidal locomotion. Journal of Theoretical Biology 242~(3),
  652--669.
\newline\urlprefix\url{http://dx.doi.org/10.1016/j.jtbi.2006.04.012}

\bibitem[{Landau and Lifschitz(1975)}]{LanLif75}
Landau, L.~D., Lifschitz, E.~M., 1975. Theory of elasticity: Transl. from the
  Russian by JB Sykes and WH Reid. Pergamon Press.

\bibitem[{Lebois et~al.(2012)Lebois, Sauvage, Py, Cardoso, Ladoux, Hersen, and
  Meglio}]{Lebois2012}
Lebois, F., Sauvage, P., Py, C., Cardoso, O., Ladoux, B., Hersen, P., Meglio,
  J.-M.~D., jun 2012. Locomotion control of caenorhabditis elegans through
  confinement. Biophysical Journal 102~(12), 2791--2798.
\newline\urlprefix\url{https://doi.org/10.1016%2Fj.bpj.2012.04.051}

\bibitem[{Lighthill(1976)}]{Lig76}
Lighthill, J., 1976. Flagellar hydrodynamics. SIAM Rev. 18~(2), 161--230.
\newline\urlprefix\url{http://dx.doi.org/10.1137/1018040}

\bibitem[{Linden(1974)}]{Linden1974}
Linden, R.~J., 1974. {Recent Advances in Physiology}. Churchill Livingston,
  Edinburgh and London.

\bibitem[{Montenegro-Johnson et~al.(2016)Montenegro-Johnson, Gagnon, Arratia,
  and Lauga}]{lauga}
Montenegro-Johnson, T.~D., Gagnon, D.~A., Arratia, P.~E., Lauga, E., Sep 2016.
  Flow analysis of the low reynolds number swimmer \textit{C.\ elegans}.
  Physical Review Fluids 1~(5).
\newline\urlprefix\url{http://dx.doi.org/10.1103/PhysRevFluids.1.053202}

\bibitem[{Niebur and Erd\"{o}s(1991)}]{Niebur1991}
Niebur, E., Erd\"{o}s, P., nov 1991. Theory of the locomotion of nematodes.
  Biophysical Journal 60~(5), 1132--1146.
\newline\urlprefix\url{http://dx.doi.org/10.1016/S0006-3495(91)82149-X}

\bibitem[{Palyanov et~al.(2016)Palyanov, Khayrulin, and Larson}]{Palyanov2016}
Palyanov, A., Khayrulin, S., Larson, S.~D., aug 2016. Application of smoothed
  particle hydrodynamics to modeling mechanisms of biological tissue. Advances
  in Engineering Software 98, 1--11.
\newline\urlprefix\url{https://doi.org/10.1016%2Fj.advengsoft.2016.03.002}

\bibitem[{Paoletti and Mahadevan(2014)}]{Paoletti2014}
Paoletti, P., Mahadevan, L., Jul 2014. A proprioceptive neuromechanical theory
  of crawling. Proceedings of the Royal Society B: Biological Sciences
  281~(1790), 20141092–20141092.
\newline\urlprefix\url{http://dx.doi.org/10.1098/rspb.2014.1092}

\bibitem[{Park et~al.(2007)Park, Goodman, and Pruitt}]{ParGooPru07}
Park, S.-J., Goodman, M.~B., Pruitt, B.~L., 2007. Analysis of nematode
  mechanics by piezoresistive displacement clamp. Proceedings of the National
  Academy of Sciences 104~(44), 17376--17381.
\newline\urlprefix\url{http://www.pnas.org/content/104/44/17376.abstract}

\bibitem[{Pierce-Shimomura et~al.(2008)Pierce-Shimomura, Chen, Mun, Ho, Sarkis,
  and McIntire}]{PieCheMun08}
Pierce-Shimomura, J.~T., Chen, B.~L., Mun, J.~J., Ho, R., Sarkis, R., McIntire,
  S.~L., 2008. {Genetic analysis of crawling and swimming locomotory patterns
  in \textit{C. elegans}}. Proceedings of the National Academy of Sciences
  105~(52), 20982--20987.

\bibitem[{Purcell(1977)}]{Pur77}
Purcell, E.~M., 1977. Life at low reynolds number. Am. J. Phys 45~(1), 3--11.

\bibitem[{Rabets et~al.(2014)Rabets, Backholm, Dalnoki-Veress, and
  Ryu}]{Rabets2014}
Rabets, Y., Backholm, M., Dalnoki-Veress, K., Ryu, W.~S., oct 2014. Direct
  measurements of drag forces in \textit{C.~elegans} crawling locomotion.
  Biophysical Journal 107~(8), 1980--1987.
\newline\urlprefix\url{http://dx.doi.org/10.1016/j.bpj.2014.09.006}

\bibitem[{Ramot et~al.(2008)Ramot, Johnson, Berry, Carnell, and
  Goodman}]{Ramot2008}
Ramot, D., Johnson, B.~E., Berry, T.~L., Carnell, L., Goodman, M.~B., may 2008.
  The parallel worm tracker: A platform for measuring average speed and
  drug-induced paralysis in nematodes. {PLoS} {ONE} 3~(5), e2208.
\newline\urlprefix\url{http://dx.doi.org/10.1371/journal.pone.0002208}

\bibitem[{Sauvage et~al.(2011)Sauvage, Argentina, Drappier, Senden,
  Sim{\'{e}}on, and Meglio}]{Sauvage2011}
Sauvage, P., Argentina, M., Drappier, J., Senden, T., Sim{\'{e}}on, J., Meglio,
  J.-M.~D., apr 2011. An elasto-hydrodynamical model of friction for the
  locomotion of caenorhabditis elegans. Journal of Biomechanics 44~(6),
  1117--1122.
\newline\urlprefix\url{https://doi.org/10.1016%2Fj.jbiomech.2011.01.026}

\bibitem[{Schulman et~al.(2014)Schulman, Backholm, Ryu, and
  Dalnoki-Veress}]{Schulman2014}
Schulman, R.~D., Backholm, M., Ryu, W.~S., Dalnoki-Veress, K., oct 2014.
  Undulatory microswimming near solid boundaries. Physics of Fluids 26~(10),
  101902.
\newline\urlprefix\url{http://dx.doi.org/10.1063/1.4897651}

\bibitem[{Stephens et~al.(2008)Stephens, Johnson-Kerner, Bialek, and
  Ryu}]{Stephens2008}
Stephens, G.~J., Johnson-Kerner, B., Bialek, W., Ryu, W.~S., apr 2008.
  Dimensionality and dynamics in the behavior of \textit{C.\ elegans}. {PLoS}
  Computational Biology 4~(4), e1000028.
\newline\urlprefix\url{https://doi.org/10.1371%2Fjournal.pcbi.1000028}

\bibitem[{Szigeti et~al.(2014)Szigeti, Gleeson, Vella, Khayrulin, Palyanov,
  Hokanson, Currie, Cantarelli, Idili, and Larson}]{Szigeti2014}
Szigeti, B., Gleeson, P., Vella, M., Khayrulin, S., Palyanov, A., Hokanson, J.,
  Currie, M., Cantarelli, M., Idili, G., Larson, S., nov 2014. {OpenWorm}: an
  open-science approach to modeling caenorhabditis elegans. Frontiers in
  Computational Neuroscience 8.
\newline\urlprefix\url{https://doi.org/10.3389%2Ffncom.2014.00137}

\bibitem[{Sznitman et~al.(2010{\natexlab{a}})Sznitman, Purohit, Krajacic,
  Lamitina, and Arratia}]{SznPurKra10}
Sznitman, J., Purohit, P.~K., Krajacic, P., Lamitina, T., Arratia, P.,
  2010{\natexlab{a}}. {Material Properties of Caenorhabditis elegans Swimming
  at Low Reynolds Number}. Biophysical Journal 98, 617--626.
\newline\urlprefix\url{http://dx.doi.org/10.1016/j.bpj.2009.11.010}

\bibitem[{Sznitman et~al.(2010{\natexlab{b}})Sznitman, Shen, Purohit, and
  Arratia}]{Sznitman2010}
Sznitman, J., Shen, X., Purohit, P.~K., Arratia, P.~E., mar 2010{\natexlab{b}}.
  The effects of fluid viscosity on the kinematics and material properties of
  \textit{C.\ elegans} swimming at low reynolds number. Experimental Mechanics
  50~(9), 1303--1311.
\newline\urlprefix\url{http://dx.doi.org/10.1007/s11340-010-9339-1}

\bibitem[{Taylor(1952)}]{Tay52}
Taylor, G., 1952. The action of waving cylindrical tails in propelling
  microscopic organisms. Proceedings of the Royal Society of London A:
  Mathematical, Physical and Engineering Sciences 211~(1105), 225--239.

\bibitem[{Thom{\'e}e(1984)}]{Tho84}
Thom{\'e}e, V., 1984. Galerkin finite element methods for parabolic problems.
  Vol. 1054. Springer.

\bibitem[{Vincent and Wegst(2004)}]{Vincent2004}
Vincent, J.~F., Wegst, U.~G., jul 2004. Design and mechanical properties of
  insect cuticle. Arthropod Structure {\&} Development 33~(3), 187--199.
\newline\urlprefix\url{https://doi.org/10.1016%2Fj.asd.2004.05.006}

\bibitem[{Vogel(1988)}]{Vog88}
Vogel, S., 1988. {Life's Devices: The physical world of animals and Plants}.
  Princeton University Press.

\bibitem[{Wallace(1969)}]{Wallace1969}
Wallace, H., jan 1969. Wave formation by infective larvae of the plant
  parasitic nematode meloidogyne javanica. Nematologica 15~(1), 65--75.
\newline\urlprefix\url{http://dx.doi.org/10.1163/187529269X00100}

\bibitem[{Wallace(1968)}]{Wallace1968}
Wallace, H.~R., sep 1968. The dynamics of nematode movement. Annual Review of
  Phytopathology 6~(1), 91--114.
\newline\urlprefix\url{https://doi.org/10.1146%2Fannurev.py.06.090168.000515}

\bibitem[{Yemini et~al.(2013)Yemini, Jucikas, Grundy, Brown, and
  Schafer}]{Yemini2013}
Yemini, E., Jucikas, T., Grundy, L.~J., Brown, A. E.~X., Schafer, W.~R., jul
  2013. A database of caenorhabditis elegans behavioral phenotypes. Nature
  Methods 10~(9), 877--879.
\newline\urlprefix\url{https://doi.org/10.1038%2Fnmeth.2560}

\end{thebibliography}

\end{document}

%%% Local Variables:
%%% mode: latex
%%% TeX-master: t
%%% End: